\documentclass[conference]{IEEEtran}

\IEEEoverridecommandlockouts

\usepackage{lineno}
\usepackage{xspace}
\usepackage{xcolor}

\def\eg{\emph{e.g.,}\xspace}

\def\ie{\emph{i.e.,}\xspace}
\def\etal{\emph{et al.}\xspace}

\usepackage{shortcuts}
\usepackage{enumerate}
\usepackage{amsthm,amsmath,amssymb,bm,mathtools,resizegather}
\usepackage{algorithm,algorithmic}
\usepackage{makecell}
\usepackage{float}
\usepackage{wasysym}
\usepackage{tcolorbox}
\usepackage{url}

\usepackage{array}
\usepackage{multirow}
\usepackage{graphicx}
\usepackage{booktabs}

\usepackage{threeparttable}
\usepackage{tikz}
\usepackage{mathtools}
\usepackage{subcaption}
\newtheorem{theorem}{Theorem}

\newtheorem{definition}{Definition}
\newtheorem*{problem*}{Problem}

\usepackage{enumitem}

\usepackage{pifont}
\newcommand{\cmark}{\ding{51}}
\newcommand{\xmark}{\ding{55}}
\usepackage{amssymb}
\usepackage{hyperref}

\def\BibTeX{{\rm B\kern-.05em{\sc i\kern-.025em b}\kern-.08em
    T\kern-.1667em\lower.7ex\hbox{E}\kern-.125emX}}

\usepackage{filecontents}
\usepackage{marvosym}

\begin{document}

\title{{\fontsize{20}{24}\selectfont Re-Key-Free, Risky-Free: Adaptable Model Usage Control}
\thanks{\textsuperscript{$\star$} Zihan Wang is supported by the Google PhD Fellowship.}
}

\author{
\IEEEauthorblockN{
Zihan Wang\textsuperscript{$\star$}\IEEEauthorrefmark{1}\IEEEauthorrefmark{2},
Zhongkui Ma\IEEEauthorrefmark{1},
Xinguo Feng\IEEEauthorrefmark{1},
Chuan Yan\IEEEauthorrefmark{1},
Dongge Liu\IEEEauthorrefmark{3}, \\[1pt]
Ruoxi Sun\IEEEauthorrefmark{2},
Derui Wang\IEEEauthorrefmark{2},
Minhui Xue\IEEEauthorrefmark{2}\IEEEauthorrefmark{4},
Guangdong Bai\IEEEauthorrefmark{5}
}\\[-8pt]
\IEEEauthorblockA{\IEEEauthorrefmark{1}\textit{The University of Queensland, Australia} \quad\IEEEauthorrefmark{2}\textit{CSIRO, Australia}\quad\IEEEauthorrefmark{3}\textit{Google LLC, Australia} }
\IEEEauthorblockA{\IEEEauthorrefmark{4}\textit{Responsible AI Research Centre, Adelaide University, Australia}\quad\IEEEauthorrefmark{5}\textit{City University of Hong Kong, China}}
}

\maketitle

\begin{abstract}
Deep neural networks (DNNs) have become valuable intellectual property of model owners, due to the substantial resources required for their development.
To protect these assets in the deployed environment, recent research has proposed model usage control mechanisms to ensure models cannot be used without proper authorization.
These methods typically lock the utility of the model by embedding an access key into its parameters.
However, they often assume static deployment, and largely fail to withstand continual post-deployment model updates, such as fine-tuning or task-specific adaptation.

In this paper, we propose \textbf{\codename}, to endow key-based model usage control with \emph{adaptability} during model evolution.
It strategically selects a subset of weights as an intrinsic access key, which enables all model updates to be confined to this key throughout the evolution lifecycle.
\codename enables using the access key to restore the keyed model to the \emph{latest} authorized states without redistributing the entire network (\ie \underline{ada}ptation), and frees the model owner from full re-keying after each model update (\ie \underline{loc}k preservation).
We establish a formal foundation to underpin \codename, providing crucial bounds such as the errors introduced by updates restricted to the access key.
Experiments across six vision and language benchmarks and six modern architectures spanning CNNs and Transformers demonstrate that \codename achieves high accuracy under significant updates while retaining robust protections.
Specifically, authorized usages consistently achieve strong task-specific performance, while unauthorized usage accuracy drops to near-random guessing levels (\eg 1.02\% on CIFAR-100), compared to up to 87.01\% under prior key-based defenses.
This shows that \codename can offer a practical solution for adaptive and protected DNN deployment in evolving real-world scenarios.
\end{abstract}

\begin{IEEEkeywords}
Intellectual property, authorization, neural networks, model adaptation.
\end{IEEEkeywords}

\section{Introduction}
\label{sec:introduction}

Major AI and cloud providers such as Google, Microsoft, and Amazon have increasingly emphasized hardware-backed execution as a foundational component of secure AI deployment. For example, Google's recent ``Private AI Compute'' places models such as Gemini within TPU-based enclaves, which harden the computational environment against external compromise~\cite{google_private_ai_compute}. Although these trusted execution infrastructures substantially improve platform security, they secure only the runtime environment and not the model artifact itself. Once a high-value model is distributed, whether through licensing, integration into downstream pipelines, or deployment on customer-managed infrastructure, a fully functional copy can leave the enclave boundary. At that point, intellectual-property control becomes structurally fragile because unrestricted copies may propagate across organizations or jurisdictions, and dishonest parties can continue operating, reselling, or embedding the model at full fidelity~\cite{zhou2023nnsplitter,sidechannel2024usenix,shuofeng24trans,wang25aim}. Such leakage compromises the model owner's competitive advantage, devalues the substantial computational and data investment embodied in the model, and exposes an inherent limitation in relying solely on hardware isolation or access control. In summary, hardware-centric protections do not mitigate misuse once the model is no longer confined to the trusted platform.

This gap indicates a fundamental requirement that goes beyond enclave-style infrastructure, namely a \textit{model-intrinsic mechanism for usage control} that can restrict a model's utility even when the hardware boundary is no longer present. Our work adopts this perspective by regulating the functional capacity of the model itself rather than relying solely on control over the execution environment. We demonstrate that hardware enclaves and model-level usage control are complementary layers. Enclaves provide secure storage and execution of secret material, while usage-control mechanisms ensure that the model remains nonfunctional without proper authorization. This layered view forms the central motivation of our work and establishes \codename as a missing component in the contemporary AI model-protection stack.

\begin{figure*}[t!]
    \centering
    \includegraphics[width=1\linewidth]{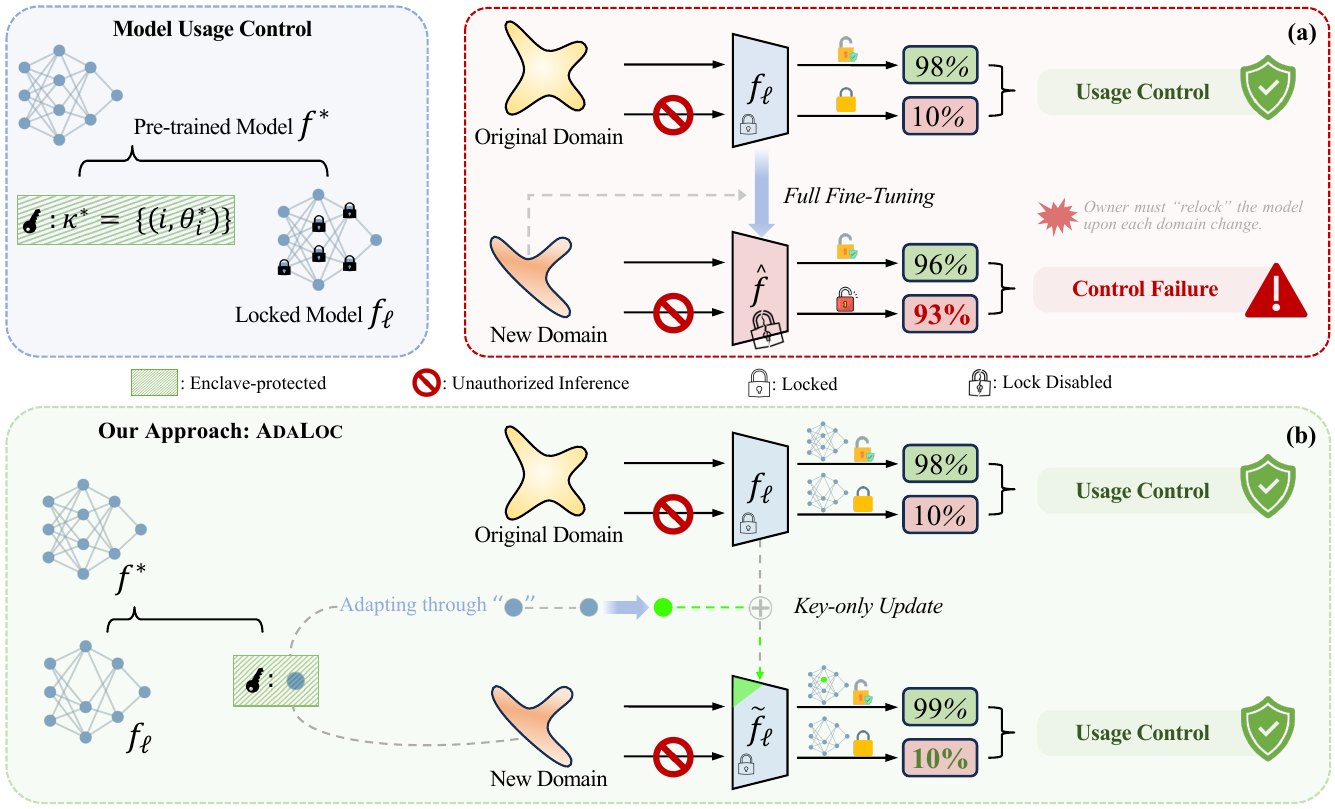}
\caption{
\textbf{General model usage control (a).}
\emph{Locking}: a pre-trained network $f^{*}$ is converted to a locked model $f_{\ell}$ plus a securely stored \emph{key} $\kappa^{*}$ (selected weight indices and values). Authorized inference returns normal outputs; unauthorized inference renders the model completely unusable (\eg random-guess level for 10-class classification task).
\emph{Update \& failure}: adapting the locked model to a new domain often invalidates the lock, restoring full access for attackers and forcing costly re-keying and model redistribution.
\textbf{Our approach (b).}
Instead of relying on any ``external'' locking mechanism, we designate a compact block of neurons themselves as the \emph{key}, and restrict all updates to that block. Because the rest of the network never changes, the lock remains intact and the updated model is still unusable without the \emph{key}.
}
\label{fig:usage_control}
\end{figure*}

Recent research has approached model usage control through \emph{model keying}~\cite{pyone2020key,chen2018piracy,Xue_2023,zhou2023nnsplitter,wang2024corelocker}.
It ensures that without the correct access key,\footnote{Throughout this paper, the term \textit{key} denotes a usage-control primitive rather than a cryptographic key.} the accuracy of the model collapses to near-random guessing, making unauthorized use effectively worthless.
Data-based keying mechanisms, such as Pyone~\etal~\cite{pyone2020key} and Chen~\etal~\cite{chen2018piracy}, embed a secret key into the training data during preprocessing, and train the model to operate only on \emph{key-preprocessed} inputs.
Model-based mechanisms, such as Xue~\etal~\cite{Xue_2023} and Zhou~\etal~\cite{zhou2023nnsplitter}, embed the key directly into internal model parameters, undermining the utility of the entire model unless the key is known by the model controller for neutralization.
These approaches mainly target the \emph{accessibility} aspect of model usage control, and have demonstrated their effectiveness in locking model utility in \emph{static} deployment.
However, they largely assume the model parameters remain unchanged after deployment.

This assumption fails to hold in practical machine learning workflows, where models are frequently fine-tuned or incrementally updated to adapt to new tasks, data distributions, or environments.
Indeed, a recent study~\cite{malec2024genai} reports that 44\% of organizations piloting AI models in 2024 relied on frequent updates, often weekly or monthly, to fine-tune established backbones~(\eg Vision Transformers~\cite{dosovitskiy2020image}).
Since even modest amounts of fine-tuning can significantly alter the parameter space, continual model refinement can easily overwrite the protection provided by existing keying schemes~(see Figure~\ref{fig:usage_control}{(a)}), making costly ``re-keying'' necessary.
This reveals a fundamental limitation in existing model keying mechanisms: \emph{accessibility} is not complemented by \emph{adaptability}.
An important research problem remains largely unexplored: \emph{how to enable model usage control under rapid and continual adaptation cycles}?

\paragraph{Our work}
We develop an adaptable model usage control framework that unifies accessibility and adaptability, named \codename~(\underline{ada}ptable \underline{loc}k).
As illustrated in Figure~\ref{fig:usage_control}{(b)}, it effectively ensures \textit{(i)} that model usability is strictly conditioned on possessing the \emph{access key}, thereby blocking unauthorized usage, and \textit{(ii)} that authorized users can efficiently update the keyed model,
eliminating the costly need to redistribute the entire network after each update.
During model updates, \codename confines the updates exclusively to the \textit{access key} (hereafter, the \emph{key}), avoiding retraining the model or altering other parts of the model.
At deployment, restoring the \emph{key} reinstates the network's core functionality, enabling robust performance on newly introduced tasks or domains.
\codename is both \emph{training data-agnostic} and \emph{retraining-free}, making it directly applicable to off-the-shelf pre-trained models and easily integrated across diverse architectures.

\codename adopts a simple yet effective strategy that leverages a minimal subset of the model's high-impact neurons as the \emph{key}.
This solution builds on the fundamental tendency known as \emph{impact concentration}~\cite{qian2021probabilistic,blalock2020state,wang2024corelocker}, where a network's predictive power consistently hinges on a small set of decisive weights.
The minimal subset ensures the \emph{key} is compact and easily manageable~(\ie \emph{key scalability}).
Through \codename, the model owner adapts the model across updates by modifying only these \emph{key} weights without touching others~(\ie \emph{effective domain adaptation}).
Removing the \emph{key} induces an accuracy drop to a completely unusable level, whereas reinserting it revives full performance in a single step (\ie \emph{restoration disparity}).
Moreover, reconstructing the missing weights is computationally infeasible~(\ie \emph{complex unauthorized restoration}).
It may take an IBM Summit supercomputer millions of years to crack a multi-layer network~\cite{chaotic2021lin}.

To establish a solid formal foundation for \codename, we for the first time formalize a comprehensive notion of model usage control that integrates accessibility and adaptability.
Accessibility ensures that a locked model $f_{\ell}$ behaves like an ideally unusable model $f^{0}$ in the absence of a valid \emph{key}~(\ie $f_{\ell}\approx f^0$), and adaptability requires this guarantee to persist after any authorized model update.
We establish crucial bounds for these two properties, which define the allowable parameter changes among the original model $f^*$, the locked model $f_{\ell}$, and the \emph{key}-only updated model $\tilde{f}_{\ell}$, in order to guarantee that the model is unusable without the \emph{key} and returns to full performance once the \emph{key} is restored.
Using layer-wise Lipschitz bounds and sub-Gaussian tail bounds, we bound the exact region in weight-space where a \emph{key}-only update can move while the network's output stays indistinguishable from full fine-tuning.
Conversely, we prove that removing the \emph{key} part collapses the network's output variance to that of an unusable reference model $f^0$.

We evaluate \codename on a wide range of image and text benchmarks using standard convolutional and transformer backbones: DenseNet-121, ResNet-152, and ConvNeXt-V2 for vision tasks, and RoBERTa, BERT, and DeBERTa for language tasks.
The evaluation addresses two questions: \textit{(i)} can a model adapt when updates are confined to the \emph{key}, and \textit{(ii)} does \codename's usage control remain effective across datasets and architectures.
Our results demonstrate that in every benchmark, the authorized model matches or even exceeds full fine-tuning accuracy with the \emph{key}, while removing the \emph{key} drives accuracy to near-random levels (\eg $1.02\%$ on CIFAR-100; see Table~\ref{tab:usage_control}).
All existing schemes still expose 76--93\% accuracy under the same setup.
These results confirm that \codename enables lightweight continual adaptation while enforcing strict usage control.

\paragraph{Contributions} We summarize our main contributions as follows:
\begin{itemize}[leftmargin=*]
\item \textbf{An adaptable usage control paradigm.} We introduce the concept of \emph{adaptable} model usage control, which enables frequent updates (\eg fine-tuning) to the keyed model while preserving the validity and functionality of the original keying mechanism. This paradigm bridges the gap between \emph{adaptability} and \emph{accessibility}, addressing scenarios where models must evolve while retaining secure usage restrictions.

\item \textbf{A solid formal foundation for usage control.} We establish a formal foundation for analyzing model usage control by, for the first time, formalizing a comprehensive notion that unifies \emph{accessibility} and \emph{adaptability}. We further derive crucial theoretical bounds that guarantee the efficacy of \codename, providing provable assurances for secure and adaptable AI model usage.

\item \textbf{An empirical evaluation.} We evaluate \codename on representative datasets and real-world models. Across all settings, \codename \textit{consistently} delivers performance fully comparable to full fine-tuning for authorized users, while models without the \emph{key} degrade to a completely unusable level. These results demonstrate that \codename enforces robust usage control while supporting seamless model adaptation across diverse datasets and architectures, making usage control practical in modern machine learning workflows.

\end{itemize}

\section{Background}
\label{sec:preliminaries}

This section introduces the foundational concepts required to understand this work.
Section~\ref{sec:dnn_definition} formalizes neural network training, and Section~\ref{sec:usage_control} formulates the problem of model usage control.

\subsection{Neural Networks}
\label{sec:dnn_definition}
To facilitate understanding, Table~\ref{tab:notation} lists the main notations used throughout the paper.
Unless otherwise stated, a symbol without superscripts denotes a generic case.

{A neural network is represented as a function $f(\bm{x}; \bm{\theta})$, where $\bm{x}\in\mathbb{R}^{d_{\mathrm{in}}}$ is the input and $\bm{\theta}\in\mathbb{R}^d$ denotes the parameters; $d_{\mathrm{in}}$ and $d$ are the input and parameter dimensions, respectively.
The choice of $\bm{\theta}$ governs the model's behavior and performance.
We use $\bm{y}^*$ to denote the target output for a given input $\bm{x}^*$.
Because our focus is on training, we assume inputs and outputs originate from the training distribution.
Below we formalize model training and model functionality.
}

\begin{table}[t]
\centering
\caption{Main notations used in this paper.}
\renewcommand{\arraystretch}{1.2}
\begin{tabular}{l p{0.75\linewidth}}
\toprule
\textbf{Notation} & \textbf{Description} \\
\midrule
$\mathcal{D}^*$
&
Original dataset. The superscript $^*$ indicates the original domain.
$(\bm{x}^*, \bm{y}^*)$ is a data-label pair from $\mathcal{D}^*$. \\
$\hat{\mathcal{D}}$ &
New dataset used for fine-tuning. $(\hat{\bm{x}}, \hat{\bm{y}})$ is a data-label pair sampled from $\hat{\mathcal{D}}$.\\
$\bm{\theta}^*$ &
Post-training parameters on $\mathcal{D}^*$, producing output $\bm{y}^*$. \\
$\hat{\bm{\theta}}$ &
Parameters obtained by \emph{full} fine-tuning on $\hat{\mathcal{D}}$. \\
$\tilde{\bm{\theta}}$ &
Parameters obtained by \emph{key-only} (partial) fine-tuning.\\
$f^{0}$ &
Ideally unusable reference model. \\
$\square_{\ell}$ &
The \emph{locked} variant of an object. \\
\bottomrule
\end{tabular}
\label{tab:notation}
\end{table}

\begin{definition}[Parameter Space]
\label{def:parameter_space}
The \emph{parameter space} of $f(\bm{x}; \bm{\theta})$ is the $d$-dimensional set $\{\bm{\theta}\in\mathbb{R}^d\}$; each point specifies a unique parameter configuration.
\end{definition}

\begin{definition}[Model Training]
\label{def:model_training}
\emph{Model training} optimizes parameters $\bm{\theta}$ to produce $\bm{\theta}^*$ by minimizing a loss function over a sequence of steps.
\end{definition}

\begin{definition}[Parameter Update]
\label{def:parameter_update}
The \emph{parameter update} is $\Delta\bm{\theta}^*=\bm{\theta}^*-\bm{\theta}$, the cumulative change applied during training.
\end{definition}

\noindent{Training halts once the loss $\mathcal{L}(f(\bm{x}^*;\bm{\theta}),\bm{y}^*)$ drops below a threshold or stops improving significantly.
The resulting parameters are typically locally optimal yet adequate for practical metrics such as accuracy or F1-score.
}

\begin{definition}[Model Functionality]
\label{def:model_functionality}
Given $\mathcal{D}^*$, a model $f(\bm{x};\bm{\theta}^*)$ satisfies \emph{functionality} if
\begin{gather}
\begin{aligned}
\mathcal{M}(f(\bm{x}^*;\bm{\theta}^*),\bm{y}^*)\le\epsilon_0,
\quad
\forall(\bm{x}^*,\bm{y}^*)\in\mathcal{D}^*,
\end{aligned}
\end{gather}
where $\mathcal{M}$ is a non-negative performance metric and $\epsilon_0$ is a small positive constant.
\end{definition}

Below we formulate the problem of model usage control, focusing on key-based mechanisms central to our approach.

\subsection{Model Usage Control}
\label{sec:usage_control}

Model usage control protects a model from unauthorized exploitation.
Commonly, a \emph{key} specifies a subset of parameters whose values gate access.

\begin{definition}[Key]
\label{def:key}
A \emph{key} $\kappa$ is a finite set of index-value pairs 
\begin{gather}
\begin{aligned}
\kappa=\{(i,v_i)\mid i\in\mathcal{S},v_i\in\mathbb{R}\},
\end{aligned}
\end{gather}
where $\mathcal{S}\subseteq\{1,\dots,d\}$ indexes the protected parameters.
\end{definition}
Each pair {specifies the value to be written back at its index on unlock}.
We denote by $\mathcal{K}$ the set of all possible keys.
{The pre-trained model corresponds to the key $\kappa^* = \{(i, \theta^*_i) \mid i \in \mathcal{S}\} \in \mathcal{K}$.}
A valid \emph{key} must satisfy the following properties:

\begin{itemize}
    \item \textbf{Compactness}: The \emph{key} covers only a minimal subset of parameters (Section~\ref{sec:usage_control:key_localization}), making it lightweight and suitable for practical deployment.
    \item \textbf{Unauthorized-use prevention}: Without the \emph{key}, the model's performance degrades and becomes unusable (Section~\ref{sec:usage_control:model_locking}).
    \item \textbf{Restoration}: With the \emph{key}, the model regains its original functionality (Section~\ref{sec:usage_control:model_unlocking}).
\end{itemize}

\subsubsection{Key Localization}
\label{sec:usage_control:key_localization}
Key localization aims to identify a small subset of model parameters that are highly sensitive to perturbations, such that modifying them alone can significantly affect model behavior. These sensitive weights form the basis for constructing an effective \emph{key}.
The goal is to select a minimal set $\mathcal{S} \subseteq \{1, \dots, d\}$, with $|\mathcal{S}| \ll d$, to ensure the resulting \emph{key} remains compact and efficient while retaining strong control over model access.

\subsubsection{Model Locking}
\label{sec:usage_control:model_locking}
To lock the model, we apply a transformation
\begin{gather}
\Phi : \mathbb{R}^d \times \mathcal{K} \to \mathbb{R}^d,
\quad
\bm{\theta}^*_\ell = \Phi(\bm{\theta}^*,\kappa^*),
\end{gather}
{which sets $\theta^*_i$ to zero for every $i \in \mathcal{S}$ and acts as the identity on the remaining coordinates.} The locked {parameter vector satisfies}
\begin{align*}
\mathcal{M}\bigl(f(\bm{x}^*;\bm{\theta}^*_\ell),\bm{y}^*\bigr) \ge M,
\end{align*}
{with $M$ chosen large enough that the model is unusable without the \emph{key}.}

\subsubsection{Model Unlocking}
\label{sec:usage_control:model_unlocking}

Unlocking reverses the locking transformation with the same key.
Let
\begin{align}
\Psi : \mathbb{R}^d \times \mathcal{K} \to \mathbb{R}^d,
\quad
\bm{\theta}^* = \Psi(\bm{\theta}^*_\ell,\kappa^*),
\end{align}
where $\Psi$ restores, for every $i\in\mathcal{S}$, {the original weight $\theta^*_i$} that was replaced during locking.
The recovered parameter vector $\bm{\theta}^*$ reinstates full model functionality, enabling normal inference. 

\section{Problem Formulation}
\label{sec:problem_formulation}

Figure \ref{fig:usage_control} exposes a critical weakness in current usage control schemes. Locking a pre-trained network $f^*$ with a \emph{key} $\kappa^*$ succeeds only while the model stays unchanged. 
Once updated (\eg fine-tuned), whether to add new features, serve a new client, or track shifting data, even small weight changes can invalidate the original lock. Unauthorized users then regain full predictive power, forcing the owner to re-key the model and repeat the costly distribution and deployment process. In rapidly evolving production environments, where models may update daily or even hourly, this break-and-rekey cycle is unsustainable. A usage-control scheme whose \emph{key} remains effective across routine updates is therefore indispensable.

\subsection{Threat Model}
\label{sec:threat_model}

This section outlines the scope of \codename and specifies the attacks it defends against, in terms of the information adversaries possess and the operations they can perform on the obtained model.

We consider a model controller who deploys a high-value neural network and may use trusted infrastructure (\eg hardware-backed enclaves or TEEs) to manage access keys during normal operation.
\codename is designed to complement such platform-level protections by ensuring that, even if the locked model itself is leaked outside these environments, it remains unusable without the \emph{key}. 

The adversary seeks to acquire a complete copy of the locked neural-network model under the control of the model controller. Possession of such a model enables various abuses, such as monetizing it through unauthorized commercial services or generating adversarial examples~\cite{goodfellow2014explaining,yuan2019adversarial} against the model owner's legitimate offerings.
The focus of our work is to endow model usage control with \emph{adaptability} during model evolution, as defined in Section~\ref{sec:introduction}, so that this protection persists across continual updates.
Access keys themselves can be managed and released by hardware-assisted mechanisms~\cite{chakraborty2020hardware} or TEE-based systems~\cite{sun2023shadownet,wang2024corelocker}, which we treat as trusted components responsible for key storage and release.

\paragraph{Adversary capabilities}
The adversary is given \emph{white-box} access to the leaked locked model: they can inspect and manipulate all weight parameters, and they know the exact network architecture used during training.
This represents a conservative assumption that favors the adversary and is realistic, since most industrial systems adopt well-published DNN designs with proven performance.
The adversary may also perform arbitrary post-processing on the leaked model, including fine-tuning, distillation, pruning, or inserting adapters, using their own compute and a limited budget of in-distribution data.
They are free to combine this data with the leaked locked model in any way they choose.

We do \emph{not} consider attacks that compromise the underlying hardware or TEE itself (\eg breaking the enclave or extracting keys from a hardware root of trust); such attacks are orthogonal to our goal and are typically addressed by the platform-security layer.
Under this threat model, \codename aims to ensure that the leaked model remains practically unusable without the \emph{key}, even under white-box access and adaptive post-processing.

\subsection{Usage Control for Adaptive AI Systems}
\label{sec:adaptable_usage_problem}

When a locked model is fine-tuned, its weights drift away from the version used to generate the original \emph{key} {(Definition~\ref{def:parameter_update})}.
If the \emph{key} is regenerated every time, deployment becomes expensive and fragile.
The challenge is therefore to maintain a \emph{single, compact key} that
\textit{(i)} continues to disable the model after \emph{any} update, and
\textit{(ii)} still restores full utility for authorized users.

Formally, let $f(\cdot;{\bm{\theta}^*_\ell})$ be {the} locked model with \emph{key} $\kappa^*$ protecting indices $\mathcal{S}$ {(Section~\ref{sec:usage_control:model_locking})}.
After full fine-tuning we obtain parameters $\hat{\bm{\theta}}$ and outputs $\hat{\bm{y}}$ on a new dataset $\hat{\mathcal{D}}$. {The post-update key $\hat{\kappa} = \{(i, \hat\theta_i) \mid i \in \mathcal{S}\}$, where $\hat\theta_i$ is the $i$-th coordinate of $\hat{\bm{\theta}}$, records the values at the same protected indices $\mathcal{S}$ as $\kappa^*$. The locked counterpart is $\hat{\bm{\theta}}_\ell = \Phi(\hat{\bm{\theta}}, \hat{\kappa})$.}
We seek a \emph{key} such that:

\begin{problem*}[Adaptable Usage Control]
\label{problem:key_usage_control}
Requiring that the performance metric $\mathcal{M}(f(\hat{\bm{x}};\hat{\bm{\theta}}),\hat{\bm{y}})\le\epsilon_0$, where $(\hat{\bm{x}}, \hat{\bm{y}}) \in \hat{\mathcal{D}}$,
we aim to find a key $\hat{\kappa}$
such that
    \begin{itemize}
        \item The model remains unusable without {the key}, i.e., $\mathcal{M}(f(\hat{\bm{x}}; \hat{\bm{\theta}}_\ell), \hat{\bm{y}}) \geq M$;
        \item The model functionality can be fully restored with {the key} presented, i.e., $\mathcal{M}(f(\hat{\bm{x}}; \Psi(\hat{\bm{\theta}}_\ell,\hat{\kappa})), \hat{\bm{y}}) \leq \epsilon_0$.
    \end{itemize}
\end{problem*}

Solving this problem is challenging, as it demands maintaining a stable \emph{key} structure while accommodating continuous model evolution.
In the following section, we introduce \codename, our adaptable usage-control framework designed specifically to meet these requirements.

\section{Our Approach: \codename}
\label{sec:approach}

Modern deep networks typically contain more parameters than necessary. Extensive pruning studies consistently reveal that a small subset of weights significantly determines model behavior, while the remaining parameters contribute minimally~\cite{li2017pruning,molchanov2016pruning,hyuntak2021pruning}. \codename leverages this intrinsic property by keeping the majority of parameters untouched, and restricting updates to a \emph{carefully selected, small subset} (designated as the \emph{key}) whenever the model needs to be fine-tuned. Specifically, the \emph{key} must satisfy two critical criteria:
\begin{enumerate}[label=(\textit{\roman*})]
\item Removal of the \emph{key} drastically reduces model accuracy.
\item Updating the \emph{key} alone achieves performance fully comparable to full fine-tuning.
\end{enumerate}

\noindent Both criteria naturally point toward selecting high-magnitude weights. Such weights significantly influence forward activations and typically accumulate larger gradient updates. Moreover, removing these influential weights generally \emph{incapacitates} the model, directly satisfying condition~\textit{(i)}. Fine-tuning restricted to this lower-dimensional subspace yields a partial-update parameter vector $\tilde{\bm{\theta}}$ that effectively steers the model towards the practical solution region near the fully tuned parameter vector $\hat{\bm{\theta}}$ (please refer to the heuristic illustration in Figure~\ref{fig:distribution}).

\begin{figure}[t]
    \centering
        \includegraphics[width=0.6\linewidth]{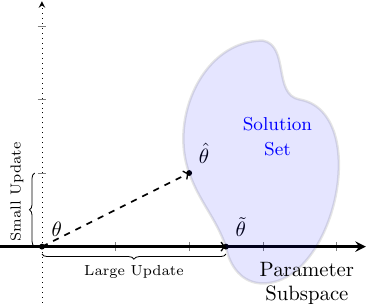}
    \caption{Heuristic illustration. Updating only a few directions with the largest parameter changes, spanning a subspace $\tilde{\bm{\theta}}$, can often reach the same solution set as full fine-tuning $\hat{\bm{\theta}}$. These high-impact directions generally align with large-magnitude weights (see Figure~\ref{fig:weight_gradient}). Theorems~\ref{thm:model_performance}--\ref{thm:small_param_gradient_conv} provide formal support for both observations, and Section~\ref{sec:experiments} empirically confirms them.
    }
    \label{fig:distribution}
\end{figure}

\subsection{Key Localization}

\codename thus adopts a straightforward yet effective heuristic by ranking parameters according to their $\ell_1$-norm and selecting the top $\rho\%$ (default $\rho = 5$) as the \emph{key}. The chosen indices form the set $\mathcal{S}$, with their initial parameter values constituting the initial adaptation \emph{key}.

The $\ell_1$-norm efficiently captures the magnitude of parameters and their relative importance in influencing network predictions. This method is computationally lightweight, inherently data-agnostic, and broadly applicable across various network architectures.
Figure~\ref{fig:horse} visually demonstrates the efficacy of this approach: filters exhibiting higher $\ell_1$-norms generally encode richer and more diverse features, whereas filters with lower magnitudes tend to capture narrower or redundant information. Thus, selecting $\mathcal{S}$ as a concise yet influential subset of parameters effectively supports impactful model adaptation.

Formally, the process of extracting the top $\rho\%$ of filters or neurons based on the $\ell_1$-norm from the $i$-th layer involves the following steps: \textit{(i)} compute the sum of absolute weights for each filter or neuron;
\textit{(ii)} rank them in descending order according to their computed sums;
\textit{(iii)} identify and remove the top $\rho\%$ of filters or neurons with the highest sum values, along with their corresponding feature maps;
\textit{(iv)} subsequently, remove filters or neurons in the following layer connected to the discarded feature maps. Filters and neurons identified and removed through this method are cataloged in the \emph{key}, clearly marking the critical parameters required for efficient and controlled model updates.

\subsection{Adapting through the Key}
\label{sec:adaption_key}

After localizing the \emph{key} with index set $\mathcal{S}$, \codename updates the model exclusively through these parameters.

Given an initial parameter vector $ \bm{\theta}^* $ and a new dataset $ \hat{\mathcal{D}} = \{(\hat{\bm{x}}^{(r)}, \hat{\bm{y}}^{(r)}) \mid 1 \leq r \leq n\} $, the \emph{key} parameters are fine-tuned by minimizing the loss function $ \mathcal{L} $ over the dataset to adapt the model to the new task. Specifically,
\begin{align*}
    \tilde{\bm{\theta}} =
    \underset{\tilde{\bm{\theta}}}{\argmin}\;
    \frac{1}{n} \sum_{r=1}^n \mathcal{L}(f(\hat{\bm{x}}^{(r)}; \tilde{\bm{\theta}}), \hat{\bm{y}}^{(r)})
    \quad \text{s.t.} \; \tilde{\theta}_j = \theta^*_j, \; \forall j \notin \mathcal{S},
\end{align*}
where $\tilde{\bm{\theta}}$ represents the partially updated parameters.

{We solve this by projected SGD restricted to coordinates in $\mathcal{S}$ with learning rate $\eta$.}
In practice, the goal is usually to meet a specified performance threshold rather than minimize the loss to absolute optimality. Thus, we define the practical solution set as the parameter configurations sufficiently close to the fully fine-tuned optimal parameters $\hat{\bm{\theta}}$ as
\begin{gather}
\begin{aligned}
    \mathcal{N}(\hat{\bm{\theta}})
    =
    \{
        \bm{\theta}' \mid \| \bm{\theta}' - \hat{\bm{\theta}} \| \leq \epsilon
    \},
\end{aligned}
\label{equ:solution}
\end{gather}
where $\epsilon > 0$ {is the parameter-space tolerance under which performance deviation from $\hat{\bm{\theta}}$ is considered acceptable}.
\begin{figure}[t]
    \centering
        \includegraphics[width=\linewidth]{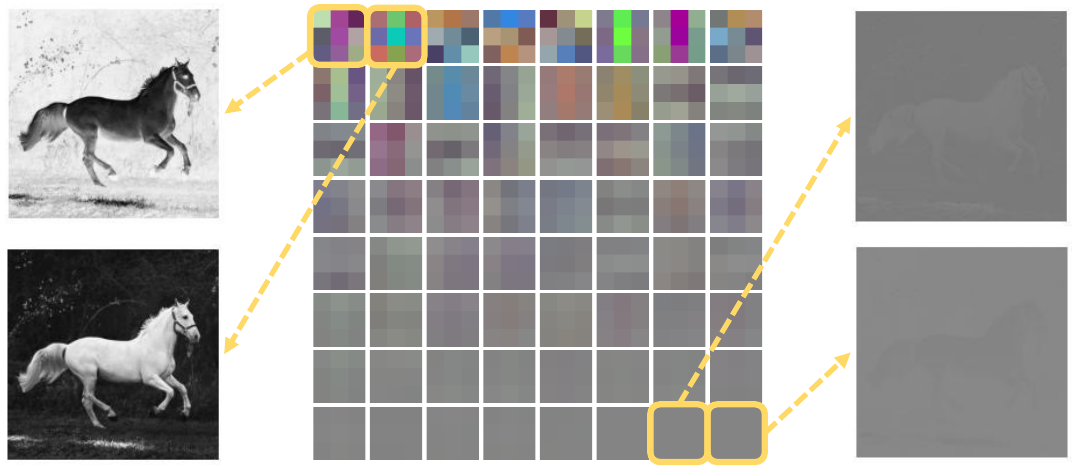}
    \caption{Visualization of feature maps (the top and bottom two) and corresponding filters (all 64 filters) from the first convolutional layer of a trained VGG model.}
    \label{fig:horse}
\end{figure}
Partial updates are effective when the adapted parameters $\tilde{\bm{\theta}}$ remain within this practical solution set, {\ie when $\|\hat{\bm{\theta}} - \tilde{\bm{\theta}}\| \le \epsilon$}. Subsequent sections theoretically (Section~\ref{sec:theory}) and empirically (Section~\ref{sec:experiments}) show that \codename's \emph{key}-updated parameters $\tilde{\bm{\theta}}$ remain within this practical solution set, demonstrating both robust authorization control and strong adaptation performance.

\begin{formal}{Remark: \textit{\codename is highly usable and deployable}}
\textit{By designating a small set of the model's own neurons as the {key}, \codename enables parameter updates to occur entirely within this subset. This allows the lock to persist across continual adaptation, satisfying the model owner's need for adaptability without sacrificing the integrity of the control mechanism. In practice, this makes it possible to share, update, and license models securely and efficiently in dynamic, real-world environments.}
\end{formal}

\section{Theoretical Analysis}
\label{sec:theory}

This section provides formal guarantees that a \emph{key} produced by \codename satisfies two essential properties: \textit{(i)} \emph{model accessibility}, ensuring the model becomes unusable without the \emph{key} (Section~\ref{sec:accessibility}), and \textit{(ii)} \emph{model adaptability}, where updating only \emph{key}-indexed weights yields performance comparable to full fine-tuning (Section~\ref{sec:adaptability}).

\paragraph{Design sketch}
The proof proceeds in four logical steps. First, Theorem~\ref{thm:variance} shows that removing the high-impact weights selected by the \emph{key} significantly collapses output variance, pushing the network toward a constant-output reference model $f^{0}$ with null utility. Next, Theorems~\ref{thm:model_performance} and~\ref{thm:upper_bound_param_diff} establish that, provided the \emph{key}-updated parameters stay within explicit distance (or standard-deviation) bounds, the model's predictions match those of a fully fine-tuned counterpart. Theorems~\ref{thm:small_param_gradient} and~\ref{thm:small_param_gradient_conv} then explain why parameters outside the \emph{key} remain stable: via H\"{o}lder's inequality, gradient magnitude at each neuron or filter is bounded proportionally to the incoming $\ell_1$-norm, so non-key parameters (which by construction have the smallest norms) receive the smallest updates and drift less during adaptation. Finally, Section~\ref{sec:theory:empirical} empirically verifies that measured parameter gaps lie well inside these theoretical bounds, confirming the practical tightness of our formal foundation.

Throughout the proofs, we adopt the following norm conventions: For vectors $\bm{x} \in \mathbb{R}^n$, $\|\bm{x}\| = \sqrt{\sum_i^n x_i^{2}}$ denotes the Euclidean norm; for matrices $\bm{A} \in \mathbb{R}^{m \times n}$, $\|\bm{A}\| = \max_{\|\bm{x}\| = 1}\|\bm{A}\bm{x}\|$ is the spectral norm; and for random variables $X$, the sub-Gaussian norm is
$\|X\|_{\varphi_{2}} = \inf\{c>0 : \mathbb{E}\exp(X^{2}/c^{2})\le 2\}$~\cite{vershynin2018high}.

\subsection{Model Accessibility}
\label{sec:accessibility}
A network is unusable if it produces almost constant outputs and cannot discriminate inputs.
We formalize this by introducing an \emph{ideally locked} reference model.

\paragraph{Reference model $f^0$}
Let $f^0$ be obtained from a pre-trained network $f^*$ by zeroing all weights:
\begin{gather}
    f^0(\bm{x})=\bm{c},\quad \forall \bm{x},
\end{gather}
where $\bm{c}$ equals the bias of the last layer.
When the test set contains uniformly distributed classes, the accuracy of $f^0$ approaches the reciprocal of the class count; in practice, its performance resembles random guessing.

The following theorem bounds output variance in terms of parameter variance, showing that shrinking a small set of high-magnitude weights drives the network toward $f^0$.

\begin{theorem}[Output variance bounded by parameter variance]
\label{thm:variance}
Let $f(\bm{x};\bm{\theta})$ be a fully connected neural network with

\begin{itemize}[leftmargin=2.2em]
    \item $L$ hidden layers and $N$ neurons per layer;
    \item activation function $\sigma$ that is Lipschitz continuous, \ie $|\sigma(a)-\sigma(b)|\le B_\sigma|a-b|$ for all $a,b\in\mathbb{R}$;
    \item weight matrices $\bm{W}^{(l)}\in\mathbb{R}^{N\times N}$ and bias vectors $\bm{b}^{(l)}\in\mathbb{R}^{N}$ for $l=1,\dots,L$.
\end{itemize}

\noindent
For the $l$-th layer, where $l = 1, \ldots, L$, assume that\footnote{Our assumptions across this study are practical and commonly used in many theoretical works on neural networks (\eg pruning~\cite{qian2021probabilistic,malach2020proving}).}

\begin{enumerate}[label=(\roman*),leftmargin=2.2em]
    \item the weights $\{W^{(l)}_{ij}\}$ are i.i.d., have zero mean, and common variance $\mathrm{Var}(\bm{W}^{(l)})$;
    \item the biases  $\{b^{(l)}_{i}\}$ are i.i.d., have zero mean, and common variance $\mathrm{Var}(\bm{b}^{(l)})$;
\end{enumerate}

\noindent
Then, for any deterministic input $\bm{x}\in\mathbb{R}^{N}$, $\mathrm{Var}(f(\bm{x} ; \bm{\theta}))$ is upper bounded by
\begin{gather}
\begin{aligned}
& \norm{\bm{x}}_2^2 (B_\sigma^2N)^L \prod_{i=1}^L \mathrm{Var}(\bm{W}^{(i)}) +
    B_\sigma^2N \cdot \mathrm{Var}(\bm{b}^{(L)})\\
    & +
    B_\sigma^2\sum_{i=1}^{L-1}\{ N \cdot \mathrm{Var}(\bm{b}^{(i)}) \prod_{j=i+1}^L [B_\sigma^2N \cdot \mathrm{Var}(\bm{W}^{(j)})]\}.
\end{aligned}
\end{gather}
\end{theorem}

\begin{proof}
We bound the output variance layer by layer, using the Lipschitz
property of~$\sigma$, and then propagate the bound forward.
For the first hidden layer, $\bm{W}^{(1)}_i\bm{x} + b^{(1)}_i$ is the linear transformation for one neuron.
Because the weights and biases are \textit{i.i.d.} with zero mean,
\begin{gather*}
\begin{array}{r @{\hspace{2pt}} l}
    \mathrm{Var}(\bm{W}^{(1)}_i\bm{x} + b^{(1)}_i)
    = & \mathrm{Var}(\sum^{N}_{j=1}W^{(1)}_{i,j}x_j+b^{(1)}_i) \\
    = & \mathrm{Var}(\sum^{N}_{j=1}W^{(1)}_{i,j}x_j) + \mathrm{Var}(b^{(1)}_i) \\
    = & \sum^{N}_{j=1}x^2_j\mathrm{Var}(W^{(1)}_{i,j}) + \mathrm{Var}(b^{(1)}_i) \\
    \leq & \mathrm{Var}(\bm{W}^{(1)})\norm{\bm{x}}^2_2 + \mathrm{Var}(\bm{b}^{(1)}).
\end{array}
\end{gather*}
Since $\sigma$ is $B_\sigma$-Lipschitz, the variance contraction property of Lipschitz maps gives $\mathrm{Var}(\sigma(Z)) \le B_\sigma^2\,\mathrm{Var}(Z)$ for any random variable $Z$. Applying this to $y^{(1)}_i = \sigma(\bm{W}^{(1)}_i\bm{x} + b^{(1)}_i)$,
\begin{gather*}
\begin{array}{r @{\hspace{2pt}} l}
\mathrm{Var}(y^{(1)}_i)
    \le & B_\sigma^2\mathrm{Var}(\bm{W}^{(1)}_i\bm{x} + b^{(1)}_i) \\
    \le & B_\sigma^2[\mathrm{Var}(\bm{W}^{(1)})\norm{\bm{x}}^2_2 + \mathrm{Var}(\bm{b}^{(1)})].
\end{array}
\end{gather*}
Summing over the $N$ neurons yields
\begin{gather*}
    \begin{array}{r @{\hspace{2pt}} l}
\mathrm{Var}(\bm{y}^{(1)})
& = \sum^{N}_{i=1}\mathrm{Var}(y^{(1)}_i) \\
& \le B_\sigma^2N[\mathrm{Var}(\bm{W}^{(1)})\norm{\bm{x}}^2_2 + \mathrm{Var}(\bm{b}^{(1)})].
    \end{array}
\end{gather*}
Similarly, for the $l$-th layer,
\begin{gather*}
    \mathrm{Var}(\bm{y}^{(l)}) \le B_\sigma^2N[\mathrm{Var}(\bm{W}^{(l)}) \norm{\bm{y}^{(l-1)}}^2 + \mathrm{Var}(\bm{b}^{(l)})],
\end{gather*}
and it concludes the theorem by iteratively applying the inequalities.
\end{proof}

\paragraph{Connecting locking to variance reduction}
Let $\alpha = \rho/100$ denote the locking fraction. Within each layer, the \emph{key} selects the $\lceil\alpha N^2\rceil$ entries of $\bm{W}^{(l)}$ with largest absolute value (the set $\mathcal{S}$ of Definition~\ref{def:key}) and zeros them. Because these entries are the top-$\alpha$ fraction by squared magnitude, they account for at least an $\alpha$-share of the total empirical second moment $\mathrm{Var}(\bm{W}^{(l)}) = \frac{1}{N^2}\sum_{i,j} (W_{ij}^{(l)})^2$. After zeroing, the residual variance therefore satisfies
\begin{gather}
\label{eq:variance_reduction}
    \mathrm{Var}(\bm{W}^{(l)}_\ell)
    \leq \frac{1}{N^2} \sum_{(i,j)\notin\mathcal{S}} (W_{ij}^{(l)})^2
    \leq (1-\alpha)\mathrm{Var}(\bm{W}^{(l)}).
\end{gather}
Substituting into Theorem~\ref{thm:variance}, whose dominant term contains $\prod_{i=1}^{L}\mathrm{Var}(\bm{W}^{(i)})$, locking across all layers scales the output variance by at most $(1-\alpha)^L$, an exponential decay in depth. In our experiments (Section~\ref{sec:experiments}), even a modest locking ratio ($\rho=5$) suffices to reduce accuracy to random-guess levels across all tested architectures.

\subsection{Model Adaptability}
\label{sec:adaptability}

We prove that updating only \codename's \emph{key}-indexed coordinates retains the accuracy that full fine-tuning would achieve.
\emph{Notation change.}
Theorem~\ref{thm:variance} required separate weight matrices $\bm{W}^{(l)}$ and bias vectors $\bm{b}^{(l)}$ because its variance bound treats them differently.
The remaining theorems reason about total parameter distance, so we merge biases into a single parameter matrix $\bm{\theta}^{(l)}$ per layer.
Throughout this subsection, we continue to use
$\hat{\bm{\theta}}$ to denote the fully fine-tuned parameter vector,
$\tilde{\bm{\theta}}$ the \emph{key-only} adaptation one, and
$\epsilon>0$ the target accuracy tolerance (the same tolerance used to define the practical solution set $\mathcal{N}(\hat{\bm{\theta}})$ in Section~\ref{sec:adaption_key}, corresponding to the threshold $\epsilon_0$ in Definition~\ref{def:model_functionality}). {The subscript $\ell$ is reserved for the locked variant; layers are indexed by the parenthesized superscript $l$.}

\subsubsection{Performance Under Small Parameter Distance}

The following theorem establishes a distance-based sufficient condition under which the \emph{key-only} updated model performs as well as the fully fine-tuned one.

\begin{theorem}[Performance under bounded parameter distance]
\label{thm:model_performance}
Let $f(\bm{x};\bm{\theta})$ be a pre-trained network and let
$\hat{\bm{\theta}}$ be a \emph{practical} solution obtained by
fine-tuning on $\hat{\mathcal{D}}$.
Write
$\bm{\theta}=\{\bm{\theta}^{(1)},\dots,\bm{\theta}^{(L)}\}$ with
$\bm{\theta}^{(l)}\in\mathbb{R}^{m^{(l)}\times n^{(l)}}$.
Assume

\begin{enumerate}[label=(\roman*),leftmargin=2.1em]
    \item The activation function $\sigma(\cdot)$ is Lipschitz continuous:
          $|\sigma(a)-\sigma(b)|\le B_\sigma|a-b|$;
    \item The weight matrices $\hat{\bm{\theta}}^{(l)}$ are bounded in spectral norm: $\| \hat{\bm{\theta}}^{(l)} \| \leq B_{\theta}$, where $B_{\theta} > 0$;
    \item Each $\tilde{\bm{\theta}}^{(l)}$ follows a sub-Gaussian distribution.
    \item Inputs obey $\|\bm{x}\|\le B_x$, where $B_x > 0$ is a constant.
\end{enumerate}

\noindent If the key update $\tilde{\bm{\theta}}$ obeys
\begin{gather}
\label{eq:param_diff}
    \| \tilde{\bm{\theta}} - \hat{\bm{\theta}} \|
    \leq \epsilon / (B_{\sigma}^{L-1} B_{\theta}^{L-1} B_x),
\end{gather}
then $\tilde{\bm{\theta}}$ falls within the practical solution set $\mathcal{N}(\hat{\bm{\theta}})$ in Eq.~\eqref{equ:solution}, such that $\| f(\bm{x}; \tilde{\bm{\theta}}) - f(\bm{x}; \hat{\bm{\theta}}) \| \leq \epsilon$ with a probability of
\begin{gather}
\label{eq:param_diff_prob}
    \Pr(\tilde{\bm{\theta}} \in \mathcal{N}(\hat{\bm{\theta}}))
    \geq (1 - 2\exp(-t^2))^L,
\end{gather}
where the term $B_{\theta} = CK^{(l)}(\sqrt{m^{(l)}} + \sqrt{n^{(l)}} +t)$, $K^{(l)} = \max_{i,j} \|\bm{\theta}^{(l)}_{i,j}\|_{\varphi_2}$, and $C$ is a universal constant.
\end{theorem}

\begin{proof}
    We bound the prediction difference between the fully fine-tuned model $f(\bm{x}; \hat{\bm{\theta}})$ and the \emph{key-only} updated model $f(\bm{x}; \tilde{\bm{\theta}})$ by recursively applying Lipschitz continuity and norm inequalities across layers. Consider
    \begin{gather*}
        f(\bm{x}; \bm{\theta}) = \bm{\theta}^{(L)} \sigma(\bm{\theta}^{(L-1)} \sigma(\cdots \sigma(\bm{\theta}^{(1)} \bm{x}))),
    \end{gather*}
    let $\bm{y}^{(l)} = \bm{\theta}^{(l)} \bm{x}^{(l)}$ for $l = 1, \ldots, L$, with {$\bm{x}^{(1)} = \bm{x}$ and} $\bm{x}^{(l)} = \sigma(\bm{y}^{(l-1)})$ for $l = {2}, \ldots, L$. Then, we have
    \begin{gather*}
    \begin{array}{r @{\hspace{2pt}} l}
        & \| f(\bm{x}; \tilde{\bm{\theta}}) - f(\bm{x}; \hat{\bm{\theta}}) \| \\
        = & \| \tilde{\bm{y}}^{(L)} - \hat{\bm{y}}^{(L)} \| \\
        = & \| \tilde{\bm{\theta}}^{(L)} \tilde{\bm{x}}^{(L)} - \hat{\bm{\theta}}^{(L)} \hat{\bm{x}}^{(L)} \| \\
        = & \| (\tilde{\bm{\theta}}^{(L)} - \hat{\bm{\theta}}^{(L)}) \hat{\bm{x}}^{(L)} + \tilde{\bm{\theta}}^{(L)} (\tilde{\bm{x}}^{(L)} - \hat{\bm{x}}^{(L)}) \| \\
        \leq & \| \tilde{\bm{\theta}}^{(L)} - \hat{\bm{\theta}}^{(L)} \| \| \hat{\bm{x}}^{(L)} \| + \| \tilde{\bm{\theta}}^{(L)} \| \| \tilde{\bm{x}}^{(L)} - \hat{\bm{x}}^{(L)} \| \\
        \leq & \| \tilde{\bm{\theta}}^{(L)} - \hat{\bm{\theta}}^{(L)} \| \| \hat{\bm{x}}^{(L)} \| + B_{\theta} \| \tilde{\bm{x}}^{(L)} - \hat{\bm{x}}^{(L)} \|.
    \end{array}
    \end{gather*}
    Since $\sigma$ is $B_\sigma$-Lipschitz and $\|\hat{\bm{\theta}}^{(l)}\|\le B_\theta$,
    \begin{gather*}
    \begin{array}{r @{\hspace{2pt}} l}
        \| \hat{\bm{x}}^{(L)} \|
        = & \| \sigma(\hat{\bm{y}}^{(L-1)}) \| \\
        \leq & B_{\sigma} \| \hat{\bm{y}}^{(L-1)} \| \\
        = & B_{\sigma} \| \hat{\bm{\theta}}^{(L-1)} \hat{\bm{x}}^{(L-1)} \| \\
        \leq & B_{\sigma} B_{\theta} \| \hat{\bm{x}}^{(L-1)} \| \\
        \leq & \cdots \\
        \leq & B_{\sigma}^{L-1} B_{\theta}^{L-1} \| \bm{x} \| \\
        \leq & B_{\sigma}^{L-1} B_{\theta}^{L-1} B_x,
    \end{array}
    \end{gather*}
    whereas for the difference of post-activations,
    \begin{gather*}
    \begin{array}{r @{\hspace{2pt}} l}
        \| \tilde{\bm{x}}^{(L)} - \hat{\bm{x}}^{(L)} \|
        = & \| \sigma(\tilde{\bm{y}}^{(L-1)}) - \sigma(\hat{\bm{y}}^{(L-1)}) \| \\
        \leq & B_{\sigma} \| \tilde{\bm{y}}^{(L-1)} - \hat{\bm{y}}^{(L-1)} \|. \\
    \end{array}
    \end{gather*}
    By repeatedly applying this nested bound across layers, we obtain
    \begin{gather*}
    \begin{array}{r @{\hspace{2pt}} l}
        & \| f(\bm{x}; \tilde{\bm{\theta}}) - f(\bm{x}; \hat{\bm{\theta}}) \| \\
        \leq & B_{\sigma}^{L-1} B_{\theta}^{L-1} B_x \| \tilde{\bm{\theta}}^{(L)} - \hat{\bm{\theta}}^{(L)} \| + B_{\sigma} B_{\theta} \| \tilde{\bm{y}}^{(L-1)} - \hat{\bm{y}}^{(L-1)} \| \\
        \leq & B_{\sigma}^{L-1} B_{\theta}^{L-1} B_x
        \sum_{l=L-1}^{L} \| \tilde{\bm{\theta}}^{(l)} - \hat{\bm{\theta}}^{(l)} \| +
        B_{\sigma}^2 B_{\theta}^2 \| \tilde{\bm{y}}^{(L-2)} - \hat{\bm{y}}^{(L-2)} \| \\
        \leq & \cdots \\
        \leq & B_{\sigma}^{L-1} B_{\theta}^{L-1} B_x
        \sum_{l=1}^{L} \| \tilde{\bm{\theta}}^{(l)} - \hat{\bm{\theta}}^{(l)} \|. \\
    \end{array}
    \end{gather*}
    Hence, demanding $\| f(\bm{x}; \tilde{\bm{\theta}}) - f(\bm{x}; \hat{\bm{\theta}}) \| \leq \epsilon$ gives
    \begin{gather*}
    \textstyle
        \| \tilde{\bm{\theta}} - \hat{\bm{\theta}} \|
        \leq \sum_{l=1}^{L} \| \tilde{\bm{\theta}}^{(l)} - \hat{\bm{\theta}}^{(l)} \|
        \leq \epsilon / (B_{\sigma}^{L-1} B_{\theta}^{L-1} B_x).
    \end{gather*}
    For sub-Gaussian weights, Vershynin's matrix concentration (\textit{Theorem 4.4.5} in~\cite{vershynin2018high}) implies that, for any $t>0$,
    \begin{gather*}
    \textstyle
        B_{\theta} = \|\bm{\theta}^{(l)}\| \leq CK^{(l)}(\sqrt{m^{(l)}} + \sqrt{n^{(l)}} +t),
    \end{gather*}
    with probability at least $1 - 2\exp(-t^2)$. Substituting this into the above bound yields the inequality~\eqref{eq:param_diff}, and the probability statement in Eq.~\eqref{eq:param_diff_prob} follows from the presence of $B_{\theta}^{L-1}$ in that expression.
\end{proof}

The bound in Theorem~\ref{thm:model_performance} is intentionally conservative. In practice, activation functions generally satisfy $B_{\sigma} \le 1$ (e.g., ReLU has $B_{\sigma}=1$), inputs are normalized so that $B_{x} \le 1$, and model parameters often have sufficiently small spectral norms to yield $B_{\theta} \le 1$. Under these conditions, the factor $(B_{\sigma} B_{\theta})^{L-1} B_{x}$ in~\eqref{eq:param_diff} decays exponentially with the network depth $L$, ensuring that $\|\tilde{\bm{\theta}} - \hat{\bm{\theta}}\|$ remains within the upper bound in typical deployments. We verify this empirically in Section~\ref{sec:theory:empirical}.

\begin{figure}
    \centering
    \includegraphics[width=1.02\linewidth]{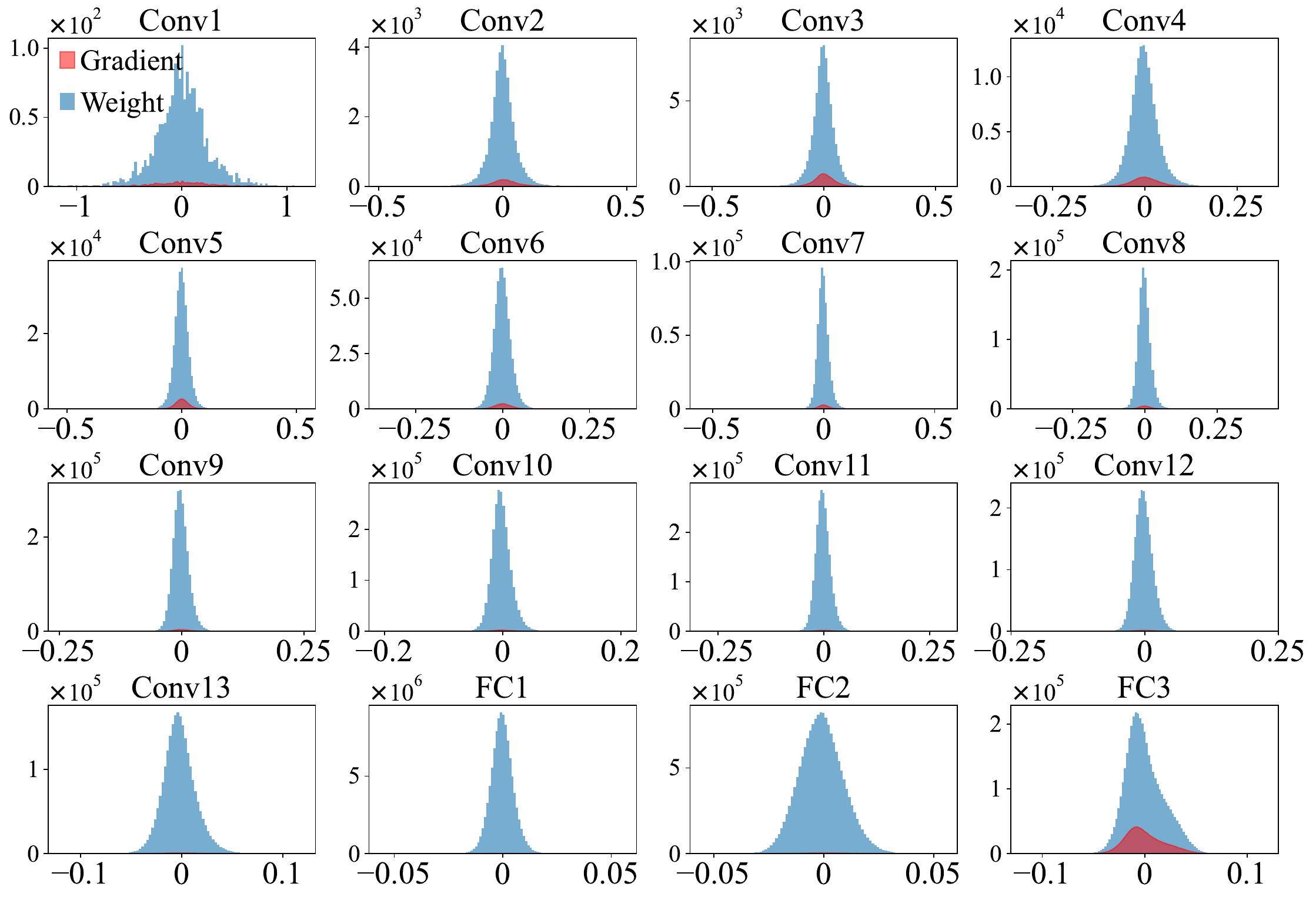}
    \caption{Layer-wise histograms of VGG-16 weights and their accumulated gradient updates during adaptation, where across all layers the vast majority of values sit at or near zero.}
    \label{fig:weight_gradient}
\end{figure}

The next theorem relaxes the requirement to one involving only the standard deviation of the parameter difference, showing that an even weaker condition still secures full performance.

\begin{theorem}[Performance under bounded standard deviation]
\label{thm:upper_bound_param_diff}
Given model parameters $\hat{\bm{\theta}}$, $\tilde{\bm{\theta}} \in \mathbb{R}^d$, where $\hat{\bm{\theta}}$ is fully updated and $\tilde{\bm{\theta}}$ is key-only (partially) updated, if
    \begin{gather}
    \label{eq:std-cond}
        \mathrm{Std}(\tilde{\bm{\theta}} - \hat{\bm{\theta}})
        \leq
        \epsilon / (B_{\sigma}^{L-1} B_{\theta}^{L-1} B_x),
    \end{gather}
    then $\tilde{\bm{\theta}}\in\mathcal{N}(\hat{\bm{\theta}})$ with probability
$(1 - 2\exp(-t^2))^{L+1}$.
\end{theorem}
\begin{proof}
Theorem 4.4.5 in~\cite{vershynin2018high} gives, for any $t>0$,
\begin{gather*}
    \| \tilde{\bm{\theta}} - \hat{\bm{\theta}} \|
    \leq
    CK_{\theta}(\sqrt{d} +t),
\end{gather*}
with a probability of at least $1 - 2\exp(-t^2)$, where $K_{\theta} = \max_{i} \| \tilde{\theta}_i - \hat{\theta}_i \|_{\varphi_2}$.
Combining this with the bound~\eqref{eq:param_diff} from
Theorem~\ref{thm:model_performance} yields
\begin{gather*}
\begin{array}{r @{\hspace{2pt}} l}
    \| \tilde{\bm{\theta}} - \hat{\bm{\theta}} \|
    \leq & CK_{\theta}(\sqrt{d} +t) \\
    \leq & \epsilon / (B_{\sigma}^{L-1} B_{\theta}^{L-1} B_x),
\end{array}
\end{gather*}
Then, we get the upper bound of the standard deviation of the parameter difference as
\begin{gather*}
\begin{array}{r @{\hspace{2pt}} l}
    \mathrm{Std}(\tilde{\bm{\theta}} - \hat{\bm{\theta}})
    \leq & \max \| \tilde{\theta}_i - \hat{\theta}_i \|_{\varphi_2} \\
    = & K_{\theta} \\
    \leq & \epsilon / [B_{\sigma}^{L-1} B_{\theta}^{L-1} B_x \cdot C(\sqrt{d} +t)].
\end{array}
\end{gather*}
\end{proof}

\begin{formal}{Remark: \textit{The bounded adaptability requirements}}
\textit{Theorems~\ref{thm:model_performance} and~\ref{thm:upper_bound_param_diff} together establish that the key-only update lies within the practical solution set $\mathcal{N}(\hat{\bm{\theta}})$ whenever the weight shift respects Eq.~\eqref{eq:param_diff}. For the architectures evaluated in Section~\ref{sec:experiments}, the empirical parameter distance consistently stays within 50--90\% of the theoretical threshold (Section~\ref{sec:theory:empirical}), confirming that the bounds are tight enough to be practically meaningful.
}
\end{formal}

Both theorems above assume that parameters outside the key remain nearly unchanged during adaptation.
\codename enforces this directly: non-key weights are frozen during key-only fine-tuning, so the distance $\|\tilde{\bm{\theta}} - \hat{\bm{\theta}}\|$ on non-key coordinates is exactly zero.
The following two theorems provide a complementary guarantee: even under full fine-tuning, gradient magnitudes on low-norm (non-key) weights remain small, bounding their drift and keeping the conditions of Theorems~\ref{thm:model_performance} and~\ref{thm:upper_bound_param_diff} satisfied.

\subsubsection{Connecting Gradient Behavior to Parameter Stability}

The two theorems below establish that gradient magnitude at each neuron (or convolutional filter) is bounded proportionally to the incoming $\ell_1$-norm of its weight vector, via a direct application of H\"{o}lder's inequality.  Consequently, neurons whose incoming weights have small $\ell_1$-norms receive small gradient updates and undergo less parameter drift during fine-tuning.

\begin{theorem}[Gradient bound for fully connected layers]
\label{thm:small_param_gradient}
Consider a fully connected neural network with ReLU activations and loss $\mathcal{L}$.  Let $\hat{\bm{\theta}}$ denote the current parameters.  For any neuron $j$ in layer $(l{-}1)$ and any neuron $i$ in layer $l$, define the backpropagated error signal
\begin{align*}
    \delta_i^{(l)}
    =
    \frac{\partial \mathcal{L}}{\partial \sigma(\bm{y}^{(l+1)})}
    \hat{\bm{\theta}}_{:,i}^{(l+1)}
    \sigma'\!\bigl(z_i^{(l)}\bigr),
\end{align*}
where $z_i^{(l)} = \hat{\bm{\theta}}_i^{(l)} \cdot \bm{y}^{(l-1)}$ is the pre-activation.  Then
\begin{align}
\label{eq:fcn_grad_bound}
    \left|\frac{\partial \mathcal{L}}{\partial \hat{\theta}_{i,j}^{(l)}}\right|
    = \bigl|\delta_i^{(l)}\bigr|
    \bigl|y_j^{(l-1)}\bigr|
    \leq \bigl|\delta_i^{(l)}\bigr|
    \bigl\|\hat{\bm{\theta}}_j^{(l-1)}\bigr\|_1
    \bigl\|\bm{y}^{(l-2)}\bigr\|_\infty.
\end{align}
Thus gradient magnitude is bounded proportionally to the incoming $\ell_1$-norm $\|\hat{\bm{\theta}}_j^{(l-1)}\|_1$.
\end{theorem}
\begin{proof}
By the chain rule,
\begin{align*}
    \frac{\partial \mathcal{L}}{\partial \hat{\theta}_{i,j}^{(l)}}
    &= \frac{\partial \mathcal{L}}{\partial \sigma(\bm{y}^{(l+1)})}
    \frac{\partial \sigma(\bm{y}^{(l+1)})}{\partial y_i^{(l)}}
    \frac{\partial y_i^{(l)}}{\partial \hat{\theta}_{i,j}^{(l)}}
    = \delta_i^{(l)} y_j^{(l-1)}.
\end{align*}
Since $y_j^{(l-1)} = \mathrm{ReLU}\bigl(\hat{\bm{\theta}}_j^{(l-1)} \cdot \bm{y}^{(l-2)}\bigr)$, we have
\begin{align*}
    \bigl|y_j^{(l-1)}\bigr|
    \leq \bigl|\hat{\bm{\theta}}_j^{(l-1)} \cdot \bm{y}^{(l-2)}\bigr|
    \leq \bigl\|\hat{\bm{\theta}}_j^{(l-1)}\bigr\|_1
    \bigl\|\bm{y}^{(l-2)}\bigr\|_\infty,
\end{align*}
where the first inequality uses $|\mathrm{ReLU}(z)| \leq |z|$ and the second is H\"{o}lder's inequality for the $(1,\infty)$ dual pair.  Substituting yields \eqref{eq:fcn_grad_bound}.
\end{proof}

The following theorem extends the same bounding technique to convolutional layers, where the inner product is replaced by a local patch convolution.
\begin{theorem}[Gradient bound for convolutional layers]
\label{thm:small_param_gradient_conv}
Consider a convolutional neural network with ReLU activations and loss $\mathcal{L}$.  Let $\hat{\bm{\theta}}$ denote the current parameters.  For input channel $s$ in layer $(l{-}1)$, the filter weights $\hat{\bm{\theta}}_{:,s}^{(l)} \in \mathbb{R}^{c_{\mathrm{out}} \times k \times k}$.  Define the backpropagated error tensor
\begin{align*}
    \delta_{c,i,j}^{(l)}
    = \frac{\partial \mathcal{L}}{\partial y_{c,i,j}^{(l)}}\,
    \sigma'\!\bigl(z_{c,i,j}^{(l)}\bigr),
\end{align*}
where $z_{c,i,j}^{(l)}$ is the pre-activation at output channel $c$ and spatial position $(i,j)$.  Then each entry of the gradient tensor satisfies
\begin{align}
\label{eq:conv_grad_bound}
    \left|\frac{\partial \mathcal{L}}{\partial \hat{\theta}_{c,s,u,v}^{(l)}}\right|
    \leq \bigl\|\bm{\delta}_{c}^{(l)}\bigr\|_1
    \bigl\|\hat{\bm{\theta}}_s^{(l-1)}\bigr\|_1
    \bigl\|\bm{y}^{(l-2)}\bigr\|_\infty,
\end{align}
where $\|\bm{\delta}_{c}^{(l)}\|_1 = \sum_{i,j}|\delta_{c,i,j}^{(l)}|$ is the spatial $\ell_1$-norm of the error signal for output channel $c$.  Thus gradient magnitude is bounded proportionally to the incoming $\ell_1$-norm $\|\hat{\bm{\theta}}_s^{(l-1)}\|_1$.
\end{theorem}
\begin{proof}
Each output feature map entry at spatial position $(i,j)$ of channel $c$ is computed as
\begin{align*}
    y_{c,i,j}^{(l)}
    = \mathrm{ReLU}\!\Bigl(
        \textstyle\sum_{s} \hat{\bm{\theta}}_{c,s}^{(l)} * y_{s,\,i:i{+}k,\,j:j{+}k}^{(l-1)}
    \Bigr),
\end{align*}
where $*$ denotes the inner product over the $k\times k$ patch.  By the chain rule, each scalar entry of the gradient tensor satisfies
\begin{align*}
    \frac{\partial \mathcal{L}}{\partial \hat{\theta}_{c,s,u,v}^{(l)}}
    = \sum_{i,j} \delta_{c,i,j}^{(l)}\, y_{s,\,i{+}u,\,j{+}v}^{(l-1)}.
\end{align*}
Since $y_{s,i,j}^{(l-1)} = \mathrm{ReLU}\bigl(\sum_{s'}\hat{\bm{\theta}}_{s',s}^{(l-1)} * \bm{y}_{s',\cdot}^{(l-2)}\bigr)$, each entry satisfies $|y_{s,i,j}^{(l-1)}| \leq \|\hat{\bm{\theta}}_s^{(l-1)}\|_1 \, \|\bm{y}^{(l-2)}\|_\infty$ by $|\mathrm{ReLU}(z)|\leq|z|$ and H\"{o}lder's inequality for the $(1,\infty)$ dual pair.  Taking absolute values and substituting,
\begin{align*}
    \left|\frac{\partial \mathcal{L}}{\partial \hat{\theta}_{c,s,u,v}^{(l)}}\right|
    &\leq \sum_{i,j} \bigl|\delta_{c,i,j}^{(l)}\bigr|
    \bigl\|\hat{\bm{\theta}}_s^{(l-1)}\bigr\|_1
    \bigl\|\bm{y}^{(l-2)}\bigr\|_\infty \\
    &= \bigl\|\bm{\delta}_{c}^{(l)}\bigr\|_1
    \bigl\|\hat{\bm{\theta}}_s^{(l-1)}\bigr\|_1
    \bigl\|\bm{y}^{(l-2)}\bigr\|_\infty,
\end{align*}
establishing~\eqref{eq:conv_grad_bound}.
\end{proof}

\begin{formal}{Remark: \textit{\codename offers strong adaptability}}
\textit{Theorems~\ref{thm:small_param_gradient} and~\ref{thm:small_param_gradient_conv} establish that gradient magnitude is bounded proportionally to the incoming $\ell_1$-norm.  Since \codename selects the highest-norm weights as the \emph{key}, the remaining (non-key) parameters have small $\ell_1$-norms and therefore receive small gradient updates at every step, accumulating less drift during fine-tuning.  Combined with the distance bounds in Theorems~\ref{thm:model_performance} and~\ref{thm:upper_bound_param_diff}, this provides formal and empirical support that \codename consistently delivers performance comparable to full fine-tuning, as confirmed across all evaluated settings (see Figure~\ref{fig:weight_gradient}).
}
\end{formal}

\subsection{Empirical Tightness of Theoretical Bounds}
\label{sec:theory:empirical}
\begin{figure}[t]
    \centering
    \includegraphics[width=\linewidth]{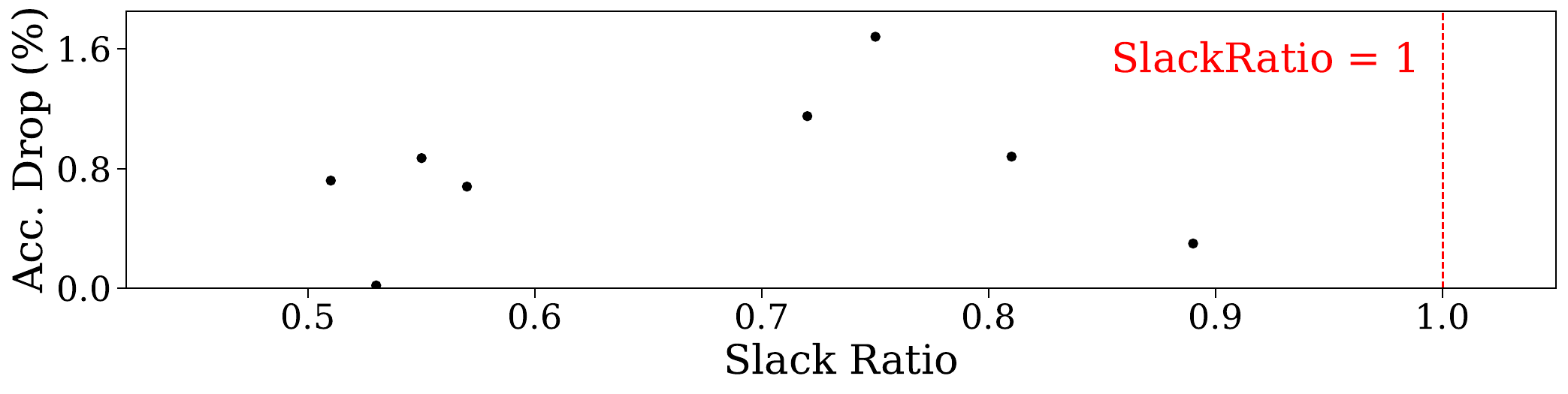}
    \caption{Accuracy drop versus \emph{slack ratio} for VGG-16 and VGG-19 on MNIST, Fashion-MNIST, CIFAR-10, and CIFAR-100, showing empirical slack below 1 while accuracy stays close to full fine-tuning.}
    \label{fig:slack}
\end{figure}

We assess the empirical tightness of Theorems~\ref{thm:model_performance} and~\ref{thm:upper_bound_param_diff} by computing the actual parameter distance $\| \tilde{\bm{\theta}} - \hat{\bm{\theta}} \|$ across various datasets.
For the theoretical bound, we compute $\epsilon / (B_\sigma^{L-1} B_\theta^{L-1} B_x)$ using constants estimated from network statistics (\eg $\mathbb{E}[B_\theta]$) following the numerical setup in~\cite{qian2021probabilistic}.
We report the resulting \emph{slack ratio} $\frac{\|\tilde{\bm{\theta}}-\hat{\bm{\theta}}\|_2}
                  {\epsilon /(B_\sigma^{L-1} B_\theta^{L-1} B_x)}$
together with accuracy drop in Figure~\ref{fig:slack} for VGG-16 and VGG-19.
For instance, on CIFAR-100 with VGG-16 the empirical distance is 0.51 while the theoretical threshold is 0.96, giving a slack ratio of 0.53.
Across all tasks, the empirical distance remained within 50--90\% of the theoretical threshold, suggesting that the bound is safe but not excessively pessimistic. Notably, models consistently retained performance within 2\% of fully fine-tuned baselines,
highlighting a consistent pattern of empirical slack.
This slack arises because the theoretical bounds are derived under worst-case assumptions, such as compounding layer-wise Lipschitz constants and uniformly tight parameter sensitivity, which overestimate actual performance degradation. In practice, sparse activation patterns, overparameterization, and flat regions in the loss landscape enable networks to absorb moderate parameter deviations with negligible performance loss. Additionally, \codename's strategy of selectively updating high-impact weights acts as an implicit regularizer, often improving generalization on smaller target datasets. Together, these factors explain why \codename delivers robust adaptation performance even beyond the strict guarantees of our theoretical bounds.

\paragraph{Summary}
The five theorems form a self-reinforcing chain: Theorem~\ref{thm:variance} guarantees that removing the \emph{key} collapses the model to an unusable state; Theorems~\ref{thm:model_performance} and~\ref{thm:upper_bound_param_diff} establish explicit distance and deviation bounds under which key-only updates match full fine-tuning; and Theorems~\ref{thm:small_param_gradient} and~\ref{thm:small_param_gradient_conv} explain \emph{why} these bounds are met by establishing that gradient magnitudes are upper-bounded proportionally to incoming $\ell_1$-norms, so that non-key parameters, which have the smallest norms, receive the smallest updates. The empirical validation above confirms that the theoretical thresholds are consistently satisfied in practice, closing the loop between the formal guarantees and observed performance.

\section{Experiments}
\label{sec:experiments}

\begin{figure}[t]
    \centering
    \includegraphics[width=1.02\linewidth]{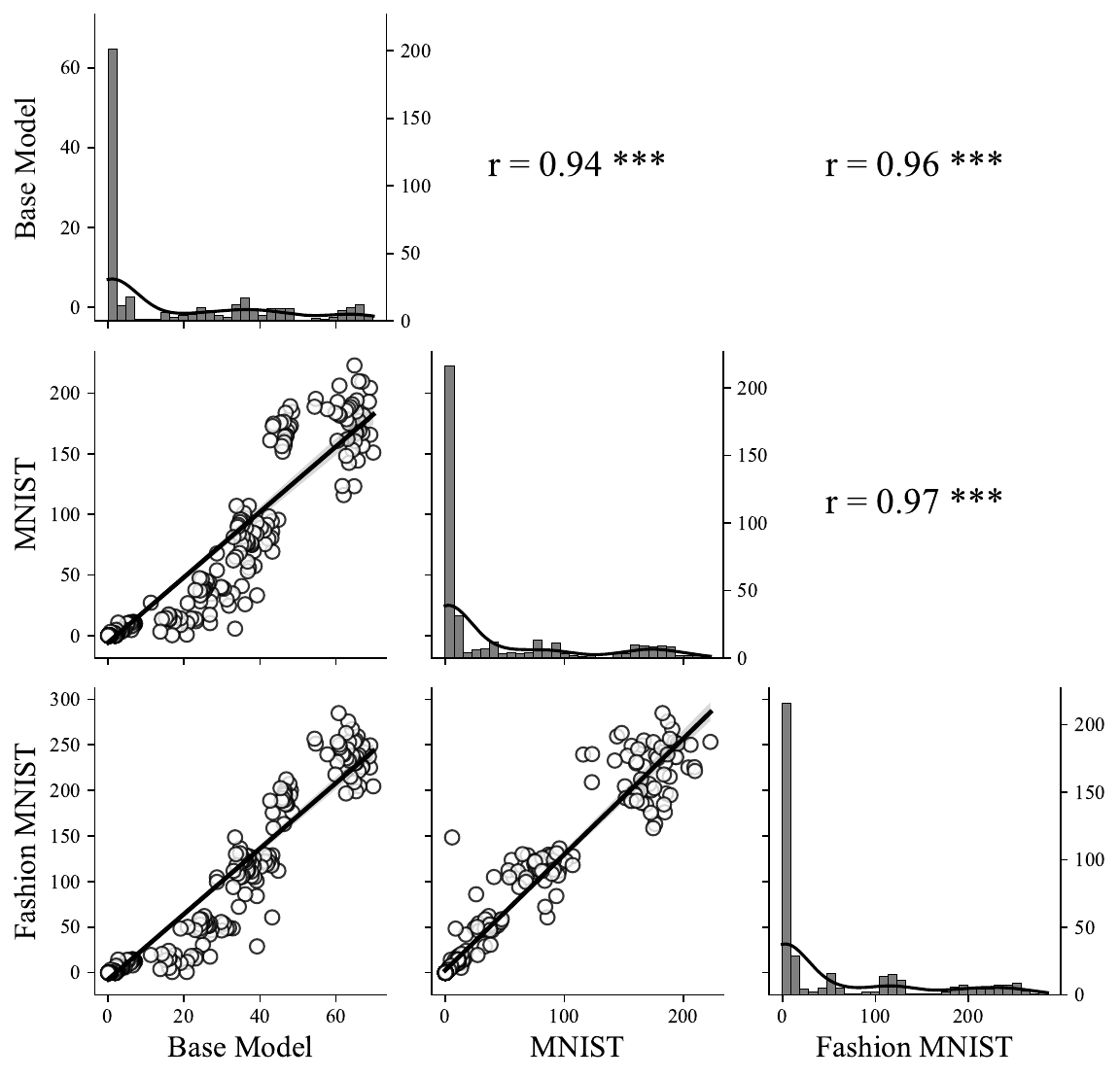}
    \caption{Visualization of the correlation between the $\ell_1$-norms of the base model (ImageNet-pre-trained) and the fine-tuned models (on MNIST and Fashion-MNIST).
    }
    \label{fig:corr}
\end{figure}

\begin{table*}[t]
\centering
\renewcommand{\arraystretch}{0.9}

\caption{Effectiveness of \codename's \emph{key} selection versus three baselines: full fine-tuning, random selection, and smallest-magnitude selection. With the \emph{key} present (authorized use), \codename  {~matches full fine-tuning accuracy and exceeds it in most settings}. Across all cells, {key removal reduces unauthorized accuracy to near-random-guess levels (qualitatively consistent with Theorem~\ref{thm:variance})}; the detailed head-to-head comparison against prior usage-control schemes is reported in Table~\ref{tab:usage_control}.}
\label{table:finetune}

\resizebox{\linewidth}{!}{
\setlength{\tabcolsep}{1pt}
\begin{threeparttable}
\scriptsize
\begin{tabular}{l rrr p{0.008\linewidth} rrr p{0.008\linewidth} rrr}
\toprule
& \multicolumn{3}{c}{CIFAR-100} && \multicolumn{3}{c}{Caltech-256} && \multicolumn{3}{c}{Flowers-102}\\
\cmidrule(l{5pt}r{5pt}){2-4} \cmidrule(l{5pt}r{5pt}){6-8} \cmidrule(l{5pt}r{5pt}){10-12}
& \makecell[c]{DenseNet-121} & \makecell[c]{ResNet-152} & \makecell[c]{ConvNeXt-V2} && \makecell[c]{DenseNet-121} & \makecell[c]{ResNet-152} & \makecell[c]{ConvNeXt-V2} && \makecell[c]{DenseNet-121} & \makecell[c]{ResNet-152} & \makecell[c]{ConvNeXt-V2} \\
\midrule
Baseline & \makecell[c]{\textbf{81.97}\%} & \makecell[c]{ 87.35\%} & \makecell[c]{88.23\%} && \makecell[c]{84.79\%} & \makecell[c]{ \textbf{85.19}\%} & \makecell[c]{92.03\%} && \makecell[c]{87.64\%} & \makecell[c]{ 87.84\%} & \makecell[c]{97.75\%}\\

Random & \makecell[c]{51.90\%} & \makecell[c]{28.57\%} & \makecell[c]{76.47\%} && \makecell[c]{38.30\%} & \makecell[c]{61.10\%} & \makecell[c]{81.82\%} && \makecell[c]{20.20\%} & \makecell[c]{~~8.33\%} & \makecell[c]{54.12\%} \\

Bottom & \makecell[c]{28.42\%} & \makecell[c]{10.31\%} & \makecell[c]{44.37\%} && \makecell[c]{16.45\%} & \makecell[c]{23.68\%} & \makecell[c]{17.67\%} && \makecell[c]{~~4.41\%} & \makecell[c]{~~1.57\%} & \makecell[c]{~~8.24\%} \\

{\codename} & \makecell[c]{81.14\%} & \makecell[c]{ \textbf{88.88}\%} & \makecell[c]{\textbf{90.17}\%} && \makecell[c]{\textbf{86.52}\%} & \makecell[c]{84.39\%} & \makecell[c]{\textbf{93.37}\%} && \makecell[c]{\,\textbf{88.14}\%} & \makecell[c]{\,\textbf{89.41}\%} & \makecell[c]{\,\textbf{98.92}\%} \\
\bottomrule
\end{tabular}

\end{threeparttable}
}
\end{table*}

In this section we proceed in three steps.
First, we verify that updating only the \emph{key} matches the accuracy of full fine-tuning on new tasks, demonstrating efficient localization and effective domain adaptation (Section \ref{sec:experiments_model_adaption}).
Then, we show that removing the \emph{key} in a static setting drops the model to random-guess levels, confirming usage control.
Finally, we compare \codename with baseline methods and highlight its ability to preserve the same control even after successive model updates (Section~\ref{sec:adaptable_usage_exp}).

\subsection{Model Adaptation}
\label{sec:experiments_model_adaption}

We begin by empirically verifying the performance guarantees discussed in the previous section. Specifically, we show that a compact \emph{key}, selected by focusing on large-magnitude parameters, consistently preserves most of the performance gains of a fully fine-tuned model.

\paragraph{Experimental setup}
We evaluate \codename on several commonly used datasets, spanning both image classification and {language classification} tasks. For image classification, we use MNIST~\cite{lecun2010mnist}, a dataset of 60,000 training and 10,000 testing grayscale images of size 28$\times$28 across 10 classes; Fashion-MNIST~\cite{xiao2017fashion}, a similar dataset with images of clothing items; CIFAR-10 and CIFAR-100~\cite{cifar10}, RGB datasets containing 50,000 training and 10,000 testing images of size 32$\times$32, where CIFAR-10 contains 10 classes and CIFAR-100 has 100 classes; Caltech-256~\cite{griffin2007caltech}, a dataset with 256 object categories for diverse classification tasks; and Flowers-102~\cite{nilsback2008automated}, a fine-grained dataset of 102 flower species. For {language classification}, we use QNLI~\cite{wang2019glue}, a question-answer inference task from the GLUE benchmark; SST-2~\cite{sst2_dataset}, a sentiment classification task; and TweetEval~\cite{barbieri-etal-2020-tweeteval}, a dataset focusing on sentiment analysis in social media text. The models evaluated include {DenseNet}, {ResNet}, and {ConvNeXt-V2} {for image classification}, and {BERT}, {RoBERTa}, and {DeBERTa} {for language classification}. All image classification models are initialized with ImageNet-pre-trained weights, while language models are initialized with their respective pre-trained checkpoints.

\begin{table*}[t]
\centering
\renewcommand{\arraystretch}{0.9}

\caption{Comparison of \codename's \emph{key-only} update and full fine-tuning on language classification tasks.}
\label{table:lm_finetune}

\resizebox{\linewidth}{!}{
\setlength{\tabcolsep}{1pt}
\begin{threeparttable}
\scriptsize
\begin{tabular}{l rrr p{0.008\linewidth} rrr p{0.008\linewidth} rrr}
\toprule
& \multicolumn{3}{c}{QNLI} && \multicolumn{3}{c}{SST-2} && \multicolumn{3}{c}{TweetEval} \\
\cmidrule(l{5pt}r{5pt}){2-4} \cmidrule(l{5pt}r{5pt}){6-8} \cmidrule(l{5pt}r{5pt}){10-12}
& \makecell[c]{RoBERTa-base} & \makecell[c]{BERT-base} & \makecell[c]{DeBERTa-base} && \makecell[c]{RoBERTa-base} & \makecell[c]{BERT-base} & \makecell[c]{DeBERTa-base} && \makecell[c]{RoBERTa-base} & \makecell[c]{BERT-base} & \makecell[c]{DeBERTa-base} \\
\midrule
Baseline & \makecell[c]{60.81\%} & \makecell[c]{61.38\%} & \makecell[c]{60.61\%} && \makecell[c]{91.74\%} & \makecell[c]{92.66\%} & \makecell[c]{94.61\%} && \makecell[c]{73.50\%} & \makecell[c]{72.80\%} & \makecell[c]{74.50\%} \\

{\codename} & \makecell[c]{{62.21}\%} & \makecell[c]{{62.68}\%} & \makecell[c]{{61.65}\%} && \makecell[c]{90.71\%} & \makecell[c]{90.56\%} & \makecell[c]{92.89\%} && \makecell[c]{67.45\%} & \makecell[c]{71.20\%} & \makecell[c]{66.20\%} \\

\bottomrule
\end{tabular}

\end{threeparttable}
}
\end{table*}

\subsubsection{Key Selection}
\label{sec:key_selection}

We analyze the weight distribution of a base model (ImageNet-pre-trained) before and after fine-tuning on MNIST and Fashion-MNIST datasets to explore the correlation between weight magnitudes and their gradient updates. Specifically, we examine how accumulated gradient updates during fine-tuning relate to the initial weight magnitudes.
The correlation graphs in Figure~\ref{fig:corr} show scatter plots, with neurons sorted by the base model neurons' $\ell_1$-norm, revealing a strong proportional relationship between the magnitude of initial neuron weights and their corresponding updates during fine-tuning. This observation supports the theoretical analysis in Section~\ref{sec:adaptability}, highlighting that large-magnitude neuron weights dominate updates during adaptation. Furthermore, the histograms (with the vertical axis on the right representing the number of neurons) demonstrate that such significant neuron weights are relatively sparse, with most parameters being near zero. These findings validate the heuristic of prioritizing large-magnitude weights for effective and efficient adaptation.

Next, we assess the fraction of weights required to achieve performance comparable to full fine-tuning. To maintain a compact \emph{key}, we tested unlocking 1\% to 10\% of the neurons with the highest magnitudes (\(\ell_1\)-norm). Our experiments show that 5\% consistently achieves strong performance across tasks, while 1-4\% suffices for simpler datasets like CIFAR-10 but underperforms on more complex tasks. Increasing beyond 5\% provides only marginal improvements. Balancing performance with key size, we adopt 5\% as the standard, as it performs reliably across all settings.

Figure~\ref{fig:ablation} illustrates two representative scenarios. For ResNet-18 on CIFAR-10 (left), high performance is observed across all tested fractions, with 5\% of neurons offering an optimal balance between accuracy and convergence speed. For DenseNet-121 on CIFAR-100 (right), an extreme case, performance drops significantly when using only 1-4\% of neurons, underscoring the robustness of selecting 5\% for domain adaptation. These results reaffirm the effectiveness of focusing updates on a small, high-impact subset of weights, ensuring efficient and reliable adaptation in diverse tasks.

\begin{figure}[t]
    \centering
    \includegraphics[width=1.02\linewidth]{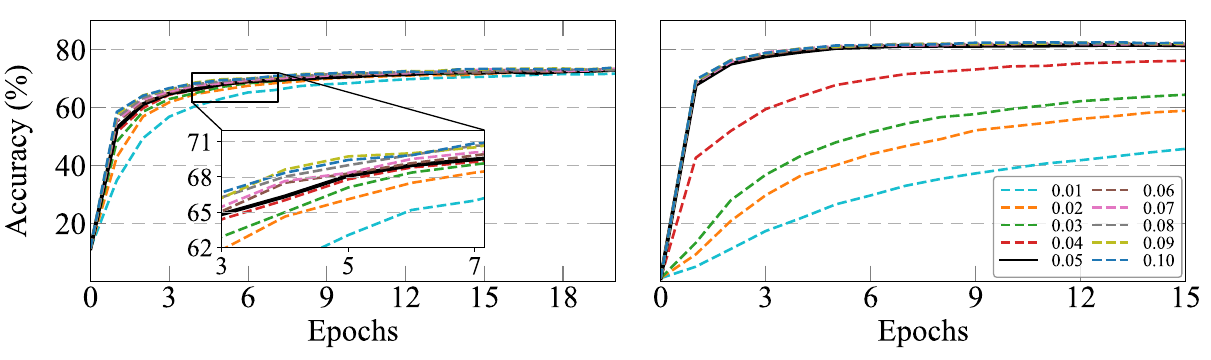}
    \caption{Training convergence of ResNet-18 on CIFAR-10 (left) and DenseNet-121 on CIFAR-100 (right) under varying \emph{key} sizes (1\%--10\% of neurons).}
\label{fig:ablation}
\end{figure}

\begin{figure*}[t]
    \centering
    \begin{subfigure}{.49\textwidth}
        \centering
        \includegraphics[width=1.01\linewidth]{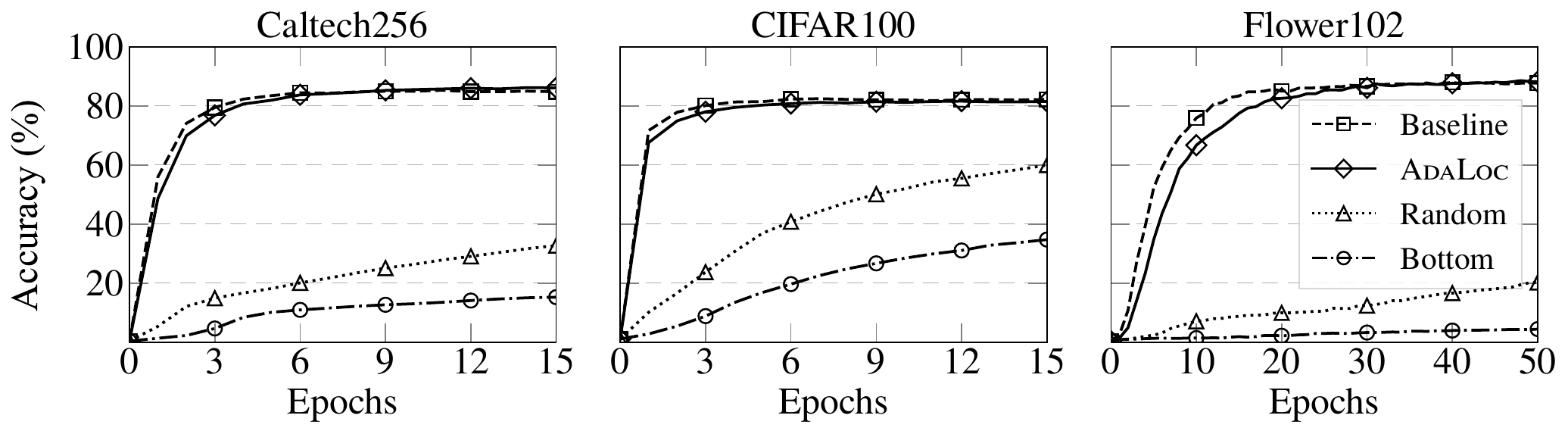}
        \label{fig:cifar10_result}
    \end{subfigure}
    \hspace{.005\textwidth}
    \begin{subfigure}{.49\textwidth}
        \centering
        \includegraphics[width=1.01\linewidth]{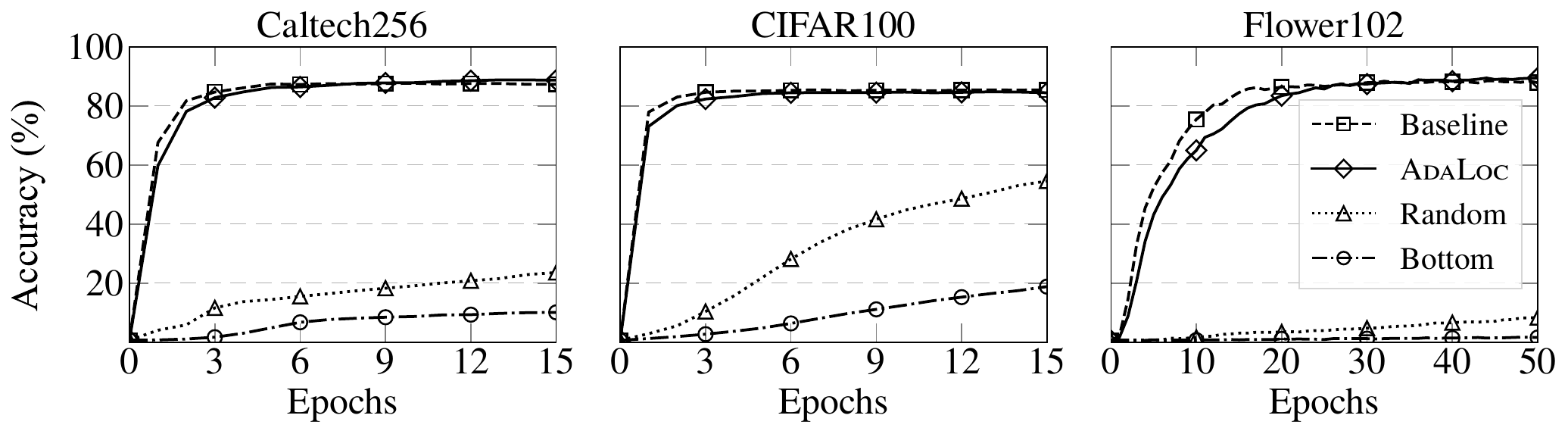}
        \label{fig:cifar100_result}
    \end{subfigure}
    \caption{Training convergence for DenseNet-121 (left) and ResNet-152 (right) with \codename \emph{key-only} updates, full fine-tuning, and two baselines (random and smallest-magnitude) across datasets.}
    \label{fig:convergence}
\end{figure*}

\begin{table*}[t]
\centering
\renewcommand{\arraystretch}{0.9}

\caption{\codename with relaxed neuron selection on {image classification}: selecting the top 5\% neurons generally yields the best accuracy, while sampling 5\% from the top 10\% or 15\% remains nearly as effective.}
\label{table:finetune_relax}

\resizebox{\linewidth}{!}{
\setlength{\tabcolsep}{1pt}
\begin{threeparttable}
\scriptsize
\begin{tabular}{l rrr p{0.008\linewidth} rrr p{0.008\linewidth} rrr}
\toprule
& \multicolumn{3}{c}{CIFAR-100} && \multicolumn{3}{c}{Caltech-256} && \multicolumn{3}{c}{Flowers-102}\\
\cmidrule(l{5pt}r{5pt}){2-4} \cmidrule(l{5pt}r{5pt}){6-8} \cmidrule(l{5pt}r{5pt}){10-12}
& \makecell[c]{DenseNet-121} & \makecell[c]{ResNet-152} & \makecell[c]{ConvNeXt-V2} && \makecell[c]{DenseNet-121} & \makecell[c]{ResNet-152} & \makecell[c]{ConvNeXt-V2} && \makecell[c]{DenseNet-121} & \makecell[c]{ResNet-152} & \makecell[c]{ConvNeXt-V2} \\
\midrule

{\codename$_{\text{top~}5\%}$} & \makecell[c]{81.14\%} & \makecell[c]{ {88.88}\%} & \makecell[c]{{90.17}\%} && \makecell[c]{{86.52}\%} & \makecell[c]{ 84.39\%} & \makecell[c]{{93.37}\%} && \makecell[c]{{88.14}\%} & \makecell[c]{{89.41}\%} & \makecell[c]{{98.92}\%} \\

\codename$_{5\%\text{~from~top~}10\%}$ & \makecell[c]{79.26\%} & \makecell[c]{80.94\%} & \makecell[c]{90.06\%} && \makecell[c]{81.06\%} & \makecell[c]{82.08\%} & \makecell[c]{92.75\%} && \makecell[c]{84.41\%} & \makecell[c]{84.61\%} & \makecell[c]{98.43\%} \\

\codename$_{5\%\text{~from~top~}15\%}$ & \makecell[c]{77.71\%} & \makecell[c]{78.31\%} & \makecell[c]{89.89\%} && \makecell[c]{79.57\%} & \makecell[c]{78.03\%} & \makecell[c]{93.25\%} && \makecell[c]{85.29\%} & \makecell[c]{80.29\%} & \makecell[c]{98.24\%} \\

\bottomrule

\end{tabular}

\end{threeparttable}
}
\end{table*}

\subsubsection{Effective Model Adaptation}
We assess adaptation under a strict 5\% parameter budget, where only 5\% of weights are used as the \emph{key}. Specifically, we compare full fine-tuning (baseline) and two alternative 5\% strategies that select neurons randomly or by small magnitude (bottom) against \codename, which prioritizes large-magnitude weights. These comparisons show that focusing updates on large-magnitude weights is markedly more effective under the same budget.

\paragraph{Results}
Table~\ref{table:finetune} presents the results of image classification across various models and datasets. The proposed approach, \codename, achieves performance comparable to full fine-tuning, attaining an accuracy of up to 98.92\%. This result reaffirms that \codename effectively identifies the critical neurons that drive model adaptation. Compared to approaches that select neurons randomly or based on small-magnitude weights, \codename delivers substantial accuracy improvements, with gains of up to 90.68 percentage points. As shown in Figure~\ref{fig:convergence}, \codename also maintains convergence rates on par with or exceeding those of full fine-tuning, confirming its efficiency and effectiveness across diverse tasks.

For language classification, a similar trend is observed as shown in Table~\ref{table:lm_finetune}, with \codename achieving an accuracy of up to 92.89\%, which is on par with full fine-tuning. This further validates \codename's effectiveness in identifying key neurons essential for model adaptation across different architectures and tasks.

In some cases, \codename even outperforms full fine-tuning. This could happen for two reasons. First, both approaches rely on finding local optima, and by focusing on a smaller subset of critical parameters, \codename may occasionally land on a better local solution. Second, for smaller fine-tuning datasets, updating fewer parameters can act as a form of regularization, helping reduce overfitting and improve generalization. This behavior underscores the advantage of selectively fine-tuning high-impact parameters.

\begin{formal2}{Takeaway}
     \textit{By tuning only the key-designated weights, chosen for their strong capacity to absorb updates, \codename reaches full fine-tuning performance for authorized users while leaving the remainder of the network untouched, as supported by Theorems~\ref{thm:model_performance}--\ref{thm:small_param_gradient_conv}.}
\end{formal2}

\begin{table*}[t]
\centering
\renewcommand{\arraystretch}{0.9}

\caption{Relaxed neuron selection on language classification shows that selecting the top 5\% neurons or sampling 5\% from the top 10\% or 15\% all yield similarly effective performance.}
\label{table:lm_finetune_relax}

\resizebox{\linewidth}{!}{
\setlength{\tabcolsep}{1pt}
\begin{threeparttable}
\scriptsize
\begin{tabular}{l rrr p{0.008\linewidth} rrr p{0.008\linewidth} rrr}
\toprule
& \multicolumn{3}{c}{QNLI} && \multicolumn{3}{c}{SST-2} && \multicolumn{3}{c}{TweetEval} \\
\cmidrule(l{5pt}r{5pt}){2-4} \cmidrule(l{5pt}r{5pt}){6-8} \cmidrule(l{5pt}r{5pt}){10-12}
& \makecell[c]{RoBERTa-base} & \makecell[c]{BERT-base} & \makecell[c]{DeBERTa-base} && \makecell[c]{RoBERTa-base} & \makecell[c]{BERT-base} & \makecell[c]{DeBERTa-base} && \makecell[c]{RoBERTa-base} & \makecell[c]{BERT-base} & \makecell[c]{DeBERTa-base} \\
\midrule

{\codename$_{\text{top}~5\%}$} & \makecell[c]{62.21\%} & \makecell[c]{62.68\%} & \makecell[c]{61.65\%} && \makecell[c]{90.71\%} & \makecell[c]{90.56\%} & \makecell[c]{92.89\%} && \makecell[c]{67.45\%} & \makecell[c]{71.20\%} & \makecell[c]{66.20\%} \\

\codename$_{5\%\text{~from~top~}10\%}$ & \makecell[c]{62.29\%} & \makecell[c]{62.29\%} & \makecell[c]{60.72\%} && \makecell[c]{90.94\%} & \makecell[c]{91.05\%} & \makecell[c]{91.97\%} && \makecell[c]{69.15\%} & \makecell[c]{71.60\%} & \makecell[c]{57.55\%} \\

\codename$_{5\%\text{~from~top~}15\%}$ & \makecell[c]{62.84\%} & \makecell[c]{62.80\%} & \makecell[c]{61.50\%} && \makecell[c]{91.28\%} & \makecell[c]{91.97\%} & \makecell[c]{92.66\%} && \makecell[c]{72.30\%} & \makecell[c]{73.55\%} & \makecell[c]{66.95\%} \\

\bottomrule

\end{tabular}

\end{threeparttable}
}
\end{table*}

\subsection{Adaptable Usage Control}
\label{sec:adaptable_usage_exp}

To evaluate the adaptability and effectiveness of \codename, we compare it with three active model usage control baselines: AdvParams \cite{Xue_2023}, NNSplitter \cite{zhou2023nnsplitter}, and CoreLocker \cite{wang2024corelocker}. Each method follows the same workflow in which a pre-trained DNN is converted into a \emph{locked model} and an authorization key that record selected weight indices and their original values. The locked model, stored in normal memory, provides only limited functionality, while the secrets are kept in a protected environment that attackers cannot reach. Once the locked model is combined with its secrets inside this secure enclave, full accuracy is restored. AdvParams injects carefully crafted adversarial perturbations into chosen weights so that, without the corrective key, the network's outputs become unusable; NNSplitter uses reinforcement learning to partition the network into a public portion and a small secret block that must be recombined for normal inference; and CoreLocker disables influential channels and relies on reinstating them to recover full performance.

\subsubsection{{Key} Pool}

To demonstrate the flexibility of \codename's \emph{key} selection, we introduce a \emph{key} pool, which relaxes the strict requirement of selecting the top 5\% of neurons by $\ell_1$ magnitude. Instead, neurons are sampled from broader subsets, specifically the top 10\% and 15\% of weights, to accommodate scenarios where precise ranking might be constrained by hardware or other mechanisms.

\paragraph{Results}
{Tables~\ref{table:finetune_relax} and~\ref{table:lm_finetune_relax} present the results of relaxed sampling. Across all models and datasets, selecting the top 5\% of neurons consistently yields the highest accuracy, reaching up to 98.92\%. Although sampling from a broader set, such as selecting 5\% from the top 10\% or 15\%, results in a slight performance degradation, it still achieves an accuracy of up to 98.43\%, closely matching the performance of the top 5\%.
}

\begin{table*}[t]
\centering
\caption{Authorized vs.\ unauthorized performance of ResNet-152 adapted from MNIST, Fashion-MNIST, CIFAR-10, and CIFAR-100. \codename consistently preserves usage control, delivering full performance for authorized users while forcing unauthorized usage to completely unusable levels.}
\label{tab:usage_control}
\renewcommand\tabcolsep{6pt}
\resizebox{\linewidth}{!}{
\begin{tabular}{lccccccccc}
\toprule
\multirow{2.6}{*}{Method} & \multicolumn{2}{c}{MNIST} & \multicolumn{2}{c}{Fashion-MNIST} & \multicolumn{2}{c}{CIFAR-10} & \multicolumn{2}{c}{CIFAR-100} & \multirow{2.6}{*}{Usage Control} \\
\cmidrule(lr){2-3}\cmidrule(lr){4-5}\cmidrule(lr){6-7}\cmidrule(lr){8-9}
 & Authorized & Unauthorized & Authorized & Unauthorized & Authorized & Unauthorized & Authorized & Unauthorized &  \\
\midrule
AdvParams~\cite{Xue_2023}           & 98.29\% & \cgreen{10.62\%} & 93.07\% & \cred{87.19\%} & 94.27\% & \cred{91.80\%} & 87.35\% & \cred{87.01\%} & \cred{\xmark} \\
NNSplitter~\cite{zhou2023nnsplitter} & 98.29\% & \cgreen{10.00\%} & 93.07\% & \cred{90.29\%} & 94.27\% & \cred{93.40\%} & 87.35\% & \cred{86.20\%} & \cred{\xmark} \\
CoreLocker~\cite{wang2024corelocker} & 98.29\% & \cgreen{10.00\%} & 93.07\% & \cred{81.71\%} & 94.27\% & \cred{80.07\%} & 87.35\% & \cred{76.26\%} & \cred{\xmark} \\
\codename                           & 98.29\% & \cgreen{10.00\%} & 92.60\% & \cgreen{10.01\%} & 91.26\% & \cgreen{10.00\%} & 88.88\% & \cgreen{~~1.02\%} & \cgreen{\cmark}\\
\bottomrule
\end{tabular}
}
\end{table*}

\begin{formal2}{Takeaway} \textit{These findings validate \textsc{AdaLoc}'s efficient localization property: even with a relaxed access-key criterion, adaptation remains nearly as effective, demonstrating robustness and practical flexibility.}
\end{formal2}

\subsubsection{Static Usage Control}

We first verify that \codename enforces usage control in a static setting, \ie without any further model updates. Theorem~\ref{thm:variance} predicts that removing the high-impact \emph{key} parameters drives the network toward the constant-output reference $f^0$, whose accuracy on a class-balanced test set {approaches $1/K$}. Across the (architecture, dataset) grid in Table~\ref{table:finetune}, key removal consistently reduces unauthorized accuracy to {near-random-guess levels, qualitatively in line with this prediction. Because per-cell values cluster tightly around chance, we defer the head-to-head comparison against prior usage-control schemes to Table~\ref{tab:usage_control}.}

To demonstrate practicality beyond convolutional backbones, we also test widely used transformer architectures. A Vision Transformer~\cite{dosovitskiy2020image} trained on CIFAR-100 falls from 81.4\% to 1.2\% when the \emph{key} is removed, and a BERT on AG News~\cite{zhang2015character} (4 classes) drops from 87.2\% to 25.3\%, effectively the random level. Because these drops occur regardless of architecture or dataset size, we reaffirm that \codename is model-agnostic: removing the key reliably degrades inference to random guessing, thereby maintaining strict static usage control.

\subsubsection{Usage Control during Model Adaptation}

We next assess whether this usage control property holds during model adaptation to new tasks. To simulate dynamic adaptation scenarios, we begin with MNIST as the base dataset where the \emph{key} is generated, and fine-tune the locked model on target datasets including Fashion-MNIST, CIFAR-10, and CIFAR-100. This setup allows us to evaluate whether \codename can maintain effective usage control while enabling legitimate updates and preventing unauthorized exploitation.

\paragraph{Results}
Table~\ref{tab:usage_control} compares \codename with three representative usage control baselines across Fashion-MNIST, CIFAR-10, and CIFAR-100. After fine-tuning, the baselines still let adversaries recover up to {93.40\%} accuracy, leaving protection almost ineffective. Under the same conditions, \codename holds unauthorized accuracy near random guess level (10.01\% on Fashion-MNIST, 10.00\% on CIFAR-10, 1.02\% on CIFAR-100), while authorized users retain high performance. The small accuracy drop for legitimate users is a worthwhile trade-off for the much stronger usage control delivered by \codename.

\subsubsection{Restoration Attack}\label{sec:experiments_restoration} {We further test whether a powerful adversary can restore the locked model. Of the adaptive vectors admitted by the threat model (Section~\ref{sec:threat_model}), only fine-tuning requires empirical evaluation; distillation, pruning, and adapter insertion inherit the static analysis of CoreLocker~\cite{wang2024corelocker}. We consider a powerful white-box adversary holding up to $20\%$ of in-distribution data, and evaluate fine-tuning recovery on Fashion-MNIST, CIFAR-10, and CIFAR-100. The worst-case recovered accuracy is {28.3\%} on Fashion-MNIST, $26.6\%$ on CIFAR-10, and {6.1\%} on CIFAR-100, against authorized accuracies of $92.60\%$, $91.26\%$, and $88.88\%$.}

\begin{formal2}{Takeaway}\textit{\codename enables continual model updates across tasks and datasets while restricting unauthorized use, {even under adaptive fine-tuning with up to $20\%$ in-distribution data.}}
\end{formal2}

\subsubsection{Discussion of Key Management}
Key management is a practical concern for \codename, since the \emph{key} must remain confidential across the model's continual updates. Two issues dominate. First, the key must be stored and transmitted without leakage; \codename integrates with trusted execution environments (TEEs), which isolate key material from untrusted memory and anchor unlock operations to hardware. Second, distributing keys to authorized users in dynamic deployments requires standard authentication, with public-key cryptography for verified retrieval and multi-factor authentication for caller verification.

Concretely, the TEE-resident payload comprises only the protected values $\{\theta^*_i\}_{i\in\mathcal{S}}$. The index set $\mathcal{S}$ itself need not be hidden: the locking transform zeros exactly the coordinates in $\mathcal{S}$ (Section~\ref{sec:usage_control:model_locking}), so an adversary with the leaked locked model can recover $\mathcal{S}$ as the support of zero entries, and protecting it inside the enclave provides no additional security. In deployment $\mathcal{S}$ can therefore reside in normal memory as a sparse coordinate map. The locking ratio is the principal knob controlling the resulting footprint: $\rho=5\%$ is adopted across architectures as the reliable setting (Section~\ref{sec:key_selection}), while smaller fractions remain effective on simpler tasks (Figure~\ref{fig:ablation}), allowing the resident size to be tuned to the available secure-memory budget. Half-precision value storage offers an additional reduction with the usual accuracy-versus-footprint trade-off.

Beyond enclave residency, \codename can be paired with hierarchical keying systems, where a master key governs subsets of the \emph{key}, to reduce distribution complexity at scale. Dynamic re-keying after each update further ensures that previously exposed keys cannot be exploited against subsequent revisions.

\begin{formal2}{Takeaway} \textit{With these key-management measures in place, every model revision stays locked to outsiders, while an authorized controller can always recover full, up-to-date accuracy by presenting the current access key.}
\end{formal2}

\section{Related Work}
\paragraph{Model IP protection}
Existing protections for deployed neural networks fall into two regimes: passive verification and active locking. Watermarking~\cite{uchida2017watermark,guo2018watermarking,kahng1998watermarking,wang2023data} supports post-hoc ownership verification but does not gate utility, so a leaked model remains fully usable. Active locking instead conditions inference on a secret. Chakraborty~\etal~\cite{chakraborty2020hardware} rely on specialized hardware and a key-dependent back-propagation process to obscure weights, which presumes a hardware-aware training pipeline and is not trivially applicable to pre-trained networks. Fan~\etal~\cite{fan2019rethinking} introduce a ``passport layer'' that preserves the network's accuracy only when the correct passport is provided. Data-based keying~\cite{chen2018piracy,pyone2020key} embeds a secret key into the training data, requiring a retraining pass for every new key. Weight-perturbation schemes inject adversarial deltas~\cite{Xue_2023}, partition the network into a secret block~\cite{zhou2023nnsplitter}, or disable influential channels~\cite{wang2024corelocker}, recovering accuracy when the secret values are written back. Across these active-locking schemes the locked parameters are treated as fixed at deployment time, so post-deployment fine-tuning shifts them and the lock degrades. \codename targets exactly this regime by confining every authorized update to the \emph{key} itself.

\paragraph{Trusted execution environment}
Trusted execution environments (TEEs), such as ARM TrustZone~\cite{ngabonziza2016trustzone,ye2018tzslicer}, have been used to deploy partitioned model variants in which a small protected portion of the computation runs inside the enclave~\cite{sun2023shadownet}, an implementation pattern we follow. The secure memory available to trusted applications is typically around 10\,MB~\cite{sun2023shadownet}, far smaller than the footprint of modern DNNs. \codename keeps only the \emph{key}-indexed weights inside the enclave, and the locking ratio $\rho$ together with reduced-precision storage (Section~\ref{sec:adaptable_usage_exp}) lets the resident size be tuned to the available budget. Side-channel attacks and other enclave-level threats~\cite{sidechannel2024usenix} are outside our scope.

\paragraph{Partial fine-tuning and weight importance}
Parameter-efficient fine-tuning updates a small subspace of weights to adapt a model while leaving the rest fixed~\cite{shen2021partial,ye2023partial}, with adapters~\cite{pan2022st}, LoRA layers~\cite{hu2021lora}, and bias-only tuning~\cite{zaken2021bitfit}. The subspace is chosen for parameter efficiency, not for the role its weights play in predictive capacity, and these methods carry no usage-control guarantee. A complementary line in pruning shows that a small subset of parameters disproportionately shapes network behavior while the rest can be removed with limited accuracy loss~\cite{li2017pruning,molchanov2016pruning,blalock2020state,qian2021probabilistic,malach2020proving}. \codename selects this same high-impact subset as the \emph{key}: removing it renders the model unusable (Theorem~\ref{thm:variance}), and authorized adaptation is then confined to the same subset. The selection criterion ($\ell_1$-magnitude) matches magnitude-based pruning, but the role is reversed; pruning discards the low-magnitude remainder, while \codename freezes it and updates only the key. A LoRA adapter cannot fill the same role because it augments the network rather than gating its native weights, so the underlying model retains usable accuracy without the adapter. Theorems~\ref{thm:model_performance}--\ref{thm:small_param_gradient_conv} bound the conditions under which key-only updates preserve the lock.

\section{Conclusion and Future Work}
We have shown that accessibility and adaptability, properties that prior keying schemes treated as conflicting, can be unified by confining every authorized update to a compact \emph{access key} drawn from a model's high-impact weights. This reframes model usage control as a property of the weight geometry rather than of the runtime that hosts it, and positions \codename as a model-intrinsic complement to enclave-based protections such as Google's Private AI Compute~\cite{google_private_ai_compute}: hardware secures the runtime, while the \emph{key} restricts the utility of any leaked artifact. The construction inverts a utility heuristic into a security primitive: pruning treats high-impact weights as the part worth keeping, while \codename treats the same locus as a withholdable secret. Empirically, authorized fine-tuning is indistinguishable from unconstrained fine-tuning on every backbone we tested, while unauthorized accuracy collapses to near-random, a separation that prior key-based defenses do not achieve.

Several limitations bound these results. Our evaluation covers discriminative vision and text models; generative and instruction-tuned regimes, where capacity is more diffuse, may admit smaller or less stable \emph{key} sets. The bounds assume layer-wise Lipschitz constants and sub-Gaussian activation tails, conditions we verify empirically but do not stress under extreme distribution shift. The threat model trusts the storage channel for the \emph{key}, and adaptive adversaries who exploit knowledge of the \emph{key}-selection rule remain outside our current analysis.

Four directions follow. Scaling \emph{key}-only adaptation to generative and large language models, where frequent fine-tuning is now standard organizational practice~\cite{malec2024genai}. Key rotation and revocation policies that preserve the lock when the access channel is compromised. Multi-stakeholder delegation, including threshold sharing of the \emph{key} indices to tolerate partial exposure or collusion. Integration with watermarking and provenance signals, so that ownership verification and usage control are enforced jointly at deployment.

\section*{Acknowledgement}
We thank Chi Wang and Zhiyong Ma for helpful discussions and experimental support. This work was supported, in part, by the UK AI Security Institute (AISI) and the Department of Industry, Science and Resources (DISR), Australia, under the Alignment Project, and by CityUHK's Start-up Grant.

\bibliographystyle{IEEEtran}
\bibliography{main}
\appendix
\section*{Ethics Considerations}
This research focuses on developing a framework for securely updating AI models under strict usage control. It relies solely on open-source datasets and publicly available pre-trained models, involves no human subjects or sensitive data, and touches no proprietary systems, ensuring compliance with ethical standards. The work is focused on theoretical analysis and experimental validation in controlled environments, without direct application to real-world services.
While \codename improves model usage control and adaptation, there is a potential risk if the \emph{key} is leaked or mismanaged. In such a scenario, unauthorized parties could bypass the intended update restrictions.
To mitigate such risks, \codename is intended for deployment in controlled environments and should integrate robust access management to prevent unauthorized usage.

\section*{Open Science Policy}
We will publicly release \codename upon publication under an open-source license. The repository is hosted at \url{https://github.com/MLresearchAI/ADALOC} and includes training and evaluation code, configuration files, and reproduction instructions, supporting community access and verification while adhering to data privacy and security guidelines.

\section*{LLM Usage Considerations}
Large language models were used solely for grammar and spelling correction, and all resulting text was manually reviewed by the authors for accuracy.

\end{document}